%% file: main.tex
\documentclass[11pt, thm-restate]{entics} 
\usepackage{enticsmacro}
\usepackage{amsmath}
\usepackage{amssymb}
\usepackage{mathrsfs}
\usepackage{macros}
\usepackage{stmaryrd}
\usepackage{color}
\usepackage{amsfonts, xcolor, colortbl}
\usepackage{mathpartir}
\usepackage{graphicx}
\usepackage{thmtools,thm-restate}
\usepackage[all]{xy}

\declaretheorem[name=Lemma,numberwithin=section]{mylemma}

\newcommand{\ra}{\rightarrow}

\def\lastname{Alves and Florido}
\begin{document}
\begin{frontmatter}
  \title{Structural Rules and Algebraic Properties of Intersection Types\thanksref{ALL}}
  \author{Sandra Alves\thanksref{myemail}}	
  \address{DCC-FCUP \& CRACS - INESCTEC, \\ University of Porto, Portugal}  							
  \thanks[ALL]{Partially funded by LIACC (FCT/UID/CEC/0027/2020)}   
   \thanks[myemail]{Email: \href{mailto:sandra@fc.up.pt} {\texttt{\normalshape
        sandra@fc.up.pt}}} 
  \author{Mário Florido\thanksref{coemail}}
  \address[a]{DCC-FCUP \& LIACC \\ University of Porto, Portugal}
  \thanks[coemail]{Email:  \href{mailto:amflorid@fc.up.pt} {\texttt{\normalshape
        amflorid@fc.up.pt}}}
\begin{abstract} 
 In this paper we define several notions of term expansion, used to define terms with less sharing, but with the same computational properties of terms typable in an intersection type system. Expansion relates terms typed by associative, commutative and idempotent intersections with terms typed in the Curry type system and the relevant type system; terms typed by non-idempotent intersections with terms typed in the affine and linear type systems; and terms typed by non-idempotent and non-commutative intersections with terms typed in an ordered type system. Finally, we show how idempotent intersection is related with the contraction rule, commutative intersection with the exchange rule and associative intersection with the lack of structural rules in a type system. 
\end{abstract}
\end{frontmatter}


\input{introduction}
\input{types}

\renewcommand{\em}{\it}
\bibliographystyle{abbrv}
\bibliography{references}

\input{appendix}

\end{document}

%% file: introduction.tex
\section{Introduction}

In the Curry Type System \cite{Curry:1934,CF:58,HS:86,Hindley:1997} each assumption about the types of free variables may be used several times or not used at all. 
This can be achieved either by considering a set of type assumptions used in the type derivation, or a list of type assumptions and the existence of three structural rules in the type system: {\em exchange}, {\em weakening} and {\em contraction}.
The  {\em exchange} rule guarantees that the order in which we write  variables in the basis is
irrelevant. The second structural rule, {\em weakening}, indicates that we  may safely add extra (unneeded) assumptions  to the basis. The third structural rule, {\em contraction}, states that if a term is typed using two identical assumptions then it is also typed using a single assumption.
 This led to the definition of a {\em substructural type system} as  type systems where the use of type assumptions is limited by the lack of one or more of the structural rules. Substructural type systems in general, have a precise relation (by the Curry-Howard correspondence) with substructural logics \cite{restall:2000}. A substructural logic is a logic where also one or more of the structural rules do not hold. Examples of well-known substructural logics include linear logic \cite{Girard:1987} and relevant logic \cite{anderson:1975}. A survey of prescriptive substructural type systems (a la Church) and its use to the control of memory resources can be found in \cite{walker}. Here we will use a descriptive view (a la Curry) of substructural type systems. Being restrictions to the Curry Type System, substructural type systems type less terms than the Curry Type System.

In the opposite direction we have Intersection Type Systems \cite{CoppoD80,CoppoDV81,bono2020tale}. They characterise exactly the set of strongly normalising terms and thus they type more terms than the Curry type system \cite{Curry:1934} or the type system of pure ML or core Haskell \cite{damas1982principal}.
 Applications of Intersection Type Systems in programming language theory cover a variety of diverse topics including the design of programming languages \cite{reynolds1997design}, program analysis \cite{PalsbergP01}, program synthesis \cite{FrankleOWZ16}, and extensions such as refinement and union types \cite{FreemanP91,DunfieldP03,Dunfield12}. But the huge expressive power of intersection types comes with a price. Type theoretic problems such as type inference and inhabitation are undecidable in general \cite{Urzyczyn94,daglib/0032840}.

In this paper we will address the following problem: to which extent can we approximate a term typed in the intersection type system by terms typable in a simpler type system, such as the Curry Type System or a substructural type system?

Let us look at the term $T \equiv (\lambda x.xx)I$, where $I$ is the
identity function $\lambda x.x$. This term has type $\alpha
\rightarrow \alpha$,  which does not involve intersections, although
it is not typable in the Curry Type System nor in any substructural type system, because it has a
non-typable subterm. The problem is the sharing of variable $x$ in $xx$, where the two shared occurrences have non-unifiable types. Now notice that
there is a term, $(\lambda x_1 x_2.x_1 x_2)II$, typable with the same type in the Curry Type System and several substructural type systems. This term uses less sharing than $(\lambda x.xx)I$ in the sense that each occurrence of a shared variable in $(\lambda x.xx)I$ corresponds to a distinct variable in $(\lambda x_1 x_2.x_1 x_2)II$.  

The present work borrows some inspiration from previous works on linearization of the $\lambda$-calculus \cite{Kfoury00,FloridoD04,AlvesF05} and aims to contribute to this line of research providing a uniform simple framework for addressing linearization related problems. 
Picking up on work in \cite{FloridoD04} we here extend the notion of term expansion to the Curry Type System and four substructural type systems: relevant, affine, linear and ordered type systems.  Under this uniform framework we  show that one can define terms with less sharing, but with the same computational properties of terms typable in an intersection type system. We will then show how we can  tune the degree of sharing by choosing different algebraic properties of the intersection operator. 

The paper is organized as follows. Section 2 presents the type systems used in the paper and sets the ground for subsequent developments. Section 3 presents several definitions of term expansion tuned by different algebraic properties of intersection types. We first use a new notion of expansion based on associative, commutative and idempotent intersection types and show how it relates intersection typed terms with terms typed in the Curry type system and the relevant type system. Then we recall previous work using non-idempotent intersection types to expand terms into linear and affine terms. Finally we use non-idempotent and non-commutative intersections and show that expanded terms are typed in an ordered type system. We further show the correcteness of all these notions of expansion in the sense that they are closed by reduction, meaning that the expansion of the normal form of a term $M$  is the normal form of the expansion of $M$. A compilation of results for different algebraic properties of intersection types can be found in the conclusions. The full version of this paper with complete proofs can be found in \cite{AlvesF22}.

%% file: types.tex
\section{Type Systems}

The type systems used here are defined for the the $\l$-calculus.
We will first recall the {\em Curry Type System} \cite{Curry:1934,CF:58} using a logic with explicit structural rules and lists of assumptions instead of the usual presentation without structural rules but using sets of assumptions.
From now on, in the rest of the paper, terms of the $\lambda$-calculus are considered module $\alpha$-equivalence and we assume that in a term $M$ no variable is bound more than once and no variable occurs both free and bound in $M$.
An infinite sequence of type-variables is assumed to be given. {\em Simple types} are expressions defined thus:
(1) each type-variable is a simple type; (2) if $\sigma$ and $\tau$ are simple types then $(\tau \rightarrow \sigma)$ is a simple type.
Type-variables are denoted by $\alpha, \beta$ and arbitrary types are denoted by $\tau, \sigma$. In both cases we may use or not number subscripts.
Parentheses will often be omitted from types, assuming that the arrow is right associative.
A finite list of pairs of the form $x:\tau$ (here called {\em assumptions}), where $x$ is a term variable and $\tau$ is a simple type, is {\em consistent} if and only if the term variables are all distinct.
A {\em basis} is a consistent finite list of pairs of the form $x:\tau$, where $x$ is a term variable and $\tau$ is a   simple type. The "," operator appends a pair to the end of the list. The list $(\Gamma_1, \Gamma_2)$ is
the list that results from appending $\Gamma_2$ onto the end of $\Gamma_1$.
We will use the notation $M:\sigma$ meaning that term $M$ has type $\sigma$ and $\Gamma \vdash M:\sigma$ to denote that $M:\sigma$ holds assuming the type declarations for free variables in the basis $\Gamma$.

The {\em Curry Type System} is defined by the following rules:
\newline

{\bf Axiom} and {\bf Structural Rules}:

\[
\inferrule{}{x:\tau \vdash x:\tau}\:{\sf (ax)}
\qquad
\inferrule{
\Gamma_1,\Gamma_2 \vdash  M:\sigma}
{\Gamma_1, x:\tau, \Gamma_2 \vdash  M:\sigma} \: {\sf (weak)}
\]
\qquad
\[
\inferrule{
\Gamma_1, x:\tau_1, y:\tau_2, \Gamma_2 \vdash  M:\sigma}
{\Gamma_1, y:\tau_2, x:\tau_1, \Gamma_2 \vdash  M:\sigma} \: {\sf (ex)}
\qquad
\inferrule{
\Gamma_1, x_1:\tau,x_2:\tau, \Gamma_2 \vdash  M:\sigma}
{\Gamma_1, x:\tau, \Gamma_2 \vdash  [x/x_1, x/x_2]M:\sigma} \: {\sf (ctr)}
\]
\newline

{\bf Logical Rules}:
\[
\inferrule{\Gamma, x:\tau_1 \vdash  M : \tau_2}{\Gamma \vdash  \l x.M : \tau_1 \rightarrow \tau_2} \:(\introarrow)
\qquad
\inferrule{\Gamma_1 \vdash  M : \tau \rightarrow \sigma  \qquad \Gamma_2 \vdash  N : \tau}{\Gamma_1 , \Gamma_2 \vdash  MN : \sigma}\;(\elimarrow)
\]

Let us describe informally the role of the structural rules.
The  Exchange rule ({\sf ex}) guarantees that the order in which we write  variables in the basis is
irrelevant. Weakening ({\sf weak}), indicates that we  may safely add extra (unneeded) assumptions  to the basis. Contraction ({\sf ctr}), states that if a term is typed using two identical assumptions then it is also typed using a single assumption.

 The lack of one or more of the structural rules leds to the definition of {\em substructural type systems}. 
There are four main substructural systems based in their logical counterparts: the {\em Relevant Type System} has only two structural rules (Exchange and Contraction). In this system assumptions are used {\em at least} once;
the {\em Affine Type System} has also two structural rules (Exchange and Weakening). In this system assumptions are used {\em at most} once; the {\em Linear Type System} has only the Exchange structural rule. In this system assumptions are used {\em exactly} once and finally, the {\em Ordered Type System} does not have any of the structural rules. In this system assumptions are used exactly once and in the {\em order} in which they are introduced in the type derivation.

\subsection{Relevant Types}

In the {\em Relevant Type System} every assumption in the basis is used to type a term. This is guaranteed by not using the Weakening type rule.
Thus the {\em Relevant Type System} corresponds to the Curry type system without the Weakening rule.

This substructural type system is related to the {\em $\lambda$I-calculus}, a restriction to the $\lambda$-calulus where in every term $M$, for each subterm of form $\lambda x.N$ in $M$, $x$ occurs free in $N$ at least once (in fact the $\lambda$I-terms were the terms that were originally studied by Church in \cite{Church:1940}).

\begin{restatable}{mylemma}{lemaa}
\label{FV}
If $\Gamma \vdash M:\tau$ is a type derivation on the {\em Relevant Type System} the set of term variables in $\Gamma$ is the set of free variables of $M$.
\end{restatable}

\begin{restatable}{thm}{theorema}
If a term $M$ is typed in the {\em Relevant Type System}, then $M$ is a $\lambda$I-term.
\end{restatable}

\subsection{Affine Types}

In the {\em Affine Type System} there is no Contraction rule. This guarantees that function parameters are used at most once. Thus the {\em Affine Type System} corresponds to the Curry type system without the Contraction rule.
The following subset of $\lambda$-terms is related to the set of terms typed in the Affine Type System.

\begin{definition}[Affine $\lambda$-terms]

An {\em affine $\lambda$-term} is a $\lambda$-term M such that:
\begin{enumerate}
\item for each subterm of $\lambda x.N$ of $M$, $x$ occurs free in $N$ at most once;
\item each free variable of $M$ has just one occurrence free in $M$.
\end{enumerate}

\end{definition}

\begin{example}
As simple examples of {\em affine} terms consider the terms $\lambda x.x$ and $\lambda x.y$. Simple examples of $\lambda$-terms which are not {\em affine} include $\lambda xyz.xz(yz)$ and $\lambda fx.f(fx)$.
\end{example}

The following theorems show that the set of terms typed in the {\em Affine Type System} is exactly the set of {\em affine terms}.

\begin{restatable}{mylemma}{FVA}
\label{FVA}
If $\Gamma \vdash M:\tau$ is a typing derivation in the {\em Affine Type System} then  the set of free variables of $M$ is included in the set of term variables in $\Gamma$.
\end{restatable}

\begin{restatable}{thm}{Affine}
A term $M$ is typed in the {\em Affine Type System}, if and only if $M$ is an {\em Affine $\lambda$-term}.
\end{restatable}

\subsection{Linear Types}
The {\em Linear Type System} corresponds to the implicational fragment of linear logic \cite{Girard:1987} confined with implication as its single connective.
In the Linear Type System each assumption must be used exactly once. This means that if $\Gamma \vdash M:\tau$ is a valid typing in the Linear Type System, then each term variable in $\Gamma$ occurs free exactly once in $M$. The Linear Type System does not have the Contraction rule, to guarantee that assumptions are used at most once, and the Weakening rule, meaning that assumptions are used exactly once.
We will follow the standard linear logic notation for functional linear types, written $\tau_1 \llto \tau_2$. 

As expected, the set of terms typed in the Linear Type System is exactly the set of {\em linear terms}.

\begin{definition}[Linear $\lambda$-terms]
A {\em linear $\lambda$-term} is a $\lambda$-term $M$ such that:
\begin{enumerate}
\item for each subterm of $\lambda x.N$ of $M$, $x$ occurs free in $N$ exactly once;
\item each free variable of $M$ has just one occurrence free in $M$.
\end{enumerate}
\end{definition}
Note that contracting a $\beta$-redex of a {\em linear term} strictly reduces the length of the term, thus linear terms have the interesting property that $\beta$-reductions starting at a linear term $M$ cannot have more contractions than the length of $M$. As a trivial consequence every linear $\lambda$-term has a normal form.

\begin{restatable}{mylemma}{FVL}
\label{FVL}
If $\Gamma \vdash M:\tau$ is a typing derivation in the {\em Linear Type System} then  the set of free variables of $M$ is equal to the set of term variables in $\Gamma$.
\end{restatable}

\begin{restatable}{thm}{LinearTerms}
A term $M$ is typed in the {\em Linear Type System}, if and only if $M$ is a {\em linear $\lambda$-term}.
\end{restatable}


\subsection{Ordered Types}
Many computational concepts are order sensitive (consider, for example, managing memory allocated on a stack). {\em Ordered Type Systems} provide a foundation for order sensitive computational problems. The central idea is that by avoiding the exchange rule, we are able to guarantee that program evaluation follows a pre-determined order. Ordered type systems are inspired by Lambek ordered logic \cite{Lambek:1958} which has several applications to natural language processing. Ordered logic was further developed by Polakow and Pfenning \cite{Polakow:1999}.

The {\em Ordered Type System} has no Contraction, thus it is linear, no Weakening, thus it is also a relevant system, and no Exchange, thus the order of use of assumptions matter. 
\begin{definition}
Let $\alpha$ range over an infinite set of type variables:
\begin{eqnarray*}
(\textbf{ordered-types})\  \tau    &::=&   \alpha \mid \tau_1 \llto_l \tau_2 \mid \tau_1 \llto_r \tau_2
\end{eqnarray*} 
\end{definition}
The definition of the {\em Ordered Type System} follows:
\newline

{\bf Axiom}:
\[
\inferrule{ }{x:\tau \vdash_o  x:\tau}\; (\ax)
\]

{\bf Logical Rules}:
\[
\inferrule{x:\tau_1,\Gamma \vdash_o  M : \tau_2}{\Gamma \vdash_o  \l x.M : \tau_1 \llto_l \tau_2} \:(\introarrowl)
\qquad
\inferrule{\Gamma,x:\tau_1 \vdash_o  M : \tau_2}{\Gamma \vdash_o  \l x.M : \tau_1 \llto_r \tau_2}\;(\introarrowr)
\]
\[
\inferrule{\Gamma_2 \vdash_o  N : \tau  \qquad \Gamma_1 \vdash_o  M : \tau \llto_l \sigma}{\Gamma_2 , \Gamma_1 \vdash_o  MN : \sigma}\;(\elimarrowl)
\qquad
\inferrule{\Gamma_1 \vdash_o  M : \tau \llto_r \sigma \qquad \Gamma_2 \vdash_o N : \tau}{\Gamma_1 , \Gamma_2 \vdash_o  MN : \sigma}\;(\elimarrowr)
\]
\newline
Note that the different arrow types guide type derivations to guarantee that the order of assumptions is used consistently during typing. For example:

\begin{example}
Consider the term $(\lambda x.xz_2)z_1$. The following two different typings are valid: 
\begin{enumerate}
    \item $z_1:\alpha \llto_r \beta, z_2:\alpha \vdash_o (\lambda x.xz_2)z_1:\beta$ 
    \item $z_2:\alpha, z_1:\alpha \llto_l \beta \vdash_o (\lambda x.xz_2)z_1:\beta$
\end{enumerate}
Note that if we change the order of the assumptions the typings are no longer valid, i.e.
\begin{enumerate}
    \item $z_2:\alpha, z_1:\alpha \llto_r \beta \not \vdash_o (\lambda x.xz_2)z_1:\beta$ 
    \item $z_1:\alpha \llto_l \beta , z_2:\alpha \not \vdash_o (\lambda x.xz_2)z_1:\beta$. 
\end{enumerate}    
\end{example}

\subsection{Intersection Types}

Intersection types originate in the works of Barendregt, Coppo and
Dezani \cite{CoppoD80,barendregt1983filter} and give us a characterization of the strongly normalizable terms. 
Consider the following example: in Intersection
Type Systems $\lambda x.xx$ has type
$(\alpha \cap \alpha \rightarrow \beta) \rightarrow \beta$. Note that the two (non-unifiable) types of the variable $x$ belong to the domain type of the abstraction linked by the intersection operator. A more
interesting example is the term $T \equiv (\lambda x.xx)I$, where $I$ is the
identity function $\lambda x.x$. This term has type $\alpha
\rightarrow \alpha$,  which does not involve intersections, although
it is not typable in the Curry Type System, because it has a
non-typable subterm.
Here we define an Intersection Type System where every type declared
in the environment is used in the type derivation, a property which is
going to be crucial in subsequent results. 

\begin{definition}
Let $\alpha$ range over an infinite set of type variables:
\begin{eqnarray*}
(\textbf{$\cap$-types})\  \sig    &::=&   \alpha  \mid \sig_1 \cap \dotsb \cap \sig_n \rightarrow \sig
\end{eqnarray*} 
\end{definition}

The original Coppo-Dezani Intersection Type System~\cite{CoppoD80} considers intersection $\cap$ as an associative, commutative and idempotent operator. There are other works which consider non idempotent intersections \cite{Kfoury00,FloridoD04,BucciarelliKV17}. 
To avoid ambiguities of notation we will use ACI-intersection to denote associative, commutative and idempotent intersections, AC-intersection to denote non-idempotent intersections and A-intersection to denote non-idempotent and non-commutative intersections. If we write only intersection we mean ACI-intersection.

\begin{definition}
A {\em type environment} is a finite set of pairs of the form $x:\tau_1 \cap \dotsb \cap \tau_n$, where $x$ is a term variable, $\tau_1 \dots \tau_n$ are  types, and the term variables are all distinct.
\end{definition}

\begin{definition}
Let $\Gamma_1$ and $\Gamma_2$ be two type environments. Then $\Gamma_1 \wedge \Gamma_2$ is the new environment given by $x:\sigma \in \Gamma_1 \wedge \Gamma_2$ if and only if $\sigma$ is defined thus
\[
\sigma = \left\{ \begin{array}{ll}
               \sigma_1 \cap \sigma_2& \mbox{if $x:\sigma_1 \in \Gamma_1$ and $x:\sigma_2 \in \Gamma_2$} \\
               \sigma_1 & \mbox{if $x:\sigma_1 \in \Gamma_1$ and $\neg
               \exists \sigma.x:\sigma \in \Gamma_2$} \\
               \sigma_2 & \mbox{if $x:\sigma_2 \in \Gamma_2$ and $\neg
               \exists \sigma.x:\sigma \in \Gamma_1$} 
            \end{array}
     \right. 
\]
\end{definition}  

The {\em Intersection Type System} used here is defined thus:
\[
\inferrule{}{\{x:\tau\} \vdash_{\cap} x:\tau}\;(\ax)
\]

\[
\inferrule{\Gamma \cup\{x:\tau_1 \cap \dotsb \cap \tau_n \} \vdash_{\cap}
M:\sigma}{\Gamma \vdash_{\cap} \lambda x.M:\tau_1 \cap \dotsb \cap \tau_n \rightarrow
\sigma}\;(\introarrow) \qquad
\inferrule{\Gamma \vdash_{\cap}
M:\sigma\quad x \not \in \fv{M}}{\Gamma \vdash_{\cap} \lambda x.M:\tau \rightarrow
\sigma}\;(\introarrows)\]

\[
\inferrule{\Gamma_0 \vdash_{\cap} M:\tau_1 \cap \dotsb \cap \tau_m \rightarrow
\sigma\qquad \left(\Gamma_i \vdash_{\cap} N:\tau_i\right)_{\iom} }{\Gamma_0 \wedge \Gamma_1 \wedge \dotsb \wedge\Gamma_m \vdash_{\cap} MN:\sigma}\;(\elimarrow)
\] 

The two different $\introarrow$ rules are necessary because in this system if there is a derivation of $\Gamma \vdash M:\sigma$ and $x$ does not occur free in $M$, then there is not a type declaration for $x$ in $\Gamma$.
The set of types for a given term $M$ in this system is strictly included in the set of types for $M$ in the original intersection type system of Coppo and Dezani \cite{CoppoD80}. For example the type $(\alpha_1 \cap \alpha_2) \rightarrow \alpha_1$ types $\lambda x.x$ in the Coppo-Dezani type system but not in the system used in this paper. The reason for this is that types in intersections for free variables can only be introduced with the $\elimarrow$ rule and thus each element of the intersection corresponds to a type that is actually used in the type derivation. However the set of terms typable in both systems is the same and corresponds to the strongly normalizable terms.

\begin{restatable}{thm}{sn}
\label{sn}
A $\lambda$-term $M$ is strongly normalizable (i.e. with no infinite reduction sequences starting from $M$) if and only if $M$ is typable in the intersection type system presented.
\end{restatable}

\section{Term Expansion}

In the following sections we will present the notion of {\em term expansion}, which generalises expansion as used in \cite{FloridoD04} to linearize the strongly normalizable terms. Under this notion we show that one can define terms with less sharing, but with the same computational properties of terms typable in an intersection type system.

Expansion consists of replacing occurrences of variable in a term, typed with different types, by a new variable typed with the corresponding types. This operation may involve other transformations in the term. For example if $x$ is expanded $k$ times in $(\lambda x.M)N$ then $N$ has to be copied $k$ times. However if the expansion is inside $N$ then $M$ may be changed, because possible arguments of $x$ may have to be copied. 

To define the expansion we face one key problem: the expansion of $MN$
is a term of the form $M_0N_1 \ldots N_k$ where $M_0$ is the expansion
of $M$ and $N_1 \ldots N_k$ are expansions of $N$. The problem here is
to find the right $k$. It is easy to determinate the number of new
arguments when $M$ is of the form $\lambda x.M'$ (just check how many fresh variables replace $x$), but if $M$ is itself an application this information depends on expansions made inside $M$. The best way to propagate this information is by using types. As we need to explicitly count the number of types of each function argument we use intersection types. If $M$ has type $\tau_1 \cap \ldots \cap \tau_k \rightarrow \sigma$ in the intersection type system we know that $MN$ will be expanded to a term of the form $M_0N_1 \ldots N_k$.

We will show that expansion relates terms typed by ACI-intersections with terms typed in the Curry Type System and the Relevant Type System, terms typed by AC-intersections with terms typed in the Affine and Linear Type Systems and terms typed by A-intersection types with terms typed in the Ordered Type System. This highlights a clear relation between algebraic properties of intersection types and the substructural rules: idempotent intersection is related with the contraction rule and commutative intersection with the exchange rule.

\subsection{From Intersection Types to Simple Types: ACI-Expansion}
\label{chapter8}
In a previous work \cite{FloridoD04} we related terms typed by non-idempotent intersections with the affine $\lambda$-calculus.
In this section we extend the results of \cite{FloridoD04} to relate terms typed by idempotent intersections with the Curry and the Relevant Type Systems.
We will define the new notion of {\em ACI-expansion} of a $\lambda$-term. 
Let us first formalize the expansion of free variables:
\begin{definition}
A {\em variable expansion} is an expression of the form
$$x:S$$
where $x$ is a variable and $S$ is a set of pairs of the form $y:\tau$ where $y$ is a variable and $\tau$ an intersection type. ($x:S$ should be read informally as ``$x$ expands to the variables in $S$''.)
\end{definition}

\begin{definition}
An {\em expansion context} $A$ is any finite set of variable expansions
$$A = \{x_1:S_1, \ldots ,x_n:S_n\}$$
where the variables $\{x_1 \ldots x_n\}$ are all different and the $S_i$ are disjoint.
\end{definition}

We now define an operation which appends two expansion contexts.
\begin{definition}

Let $A_1$ and $A_2$ be two expansion contexts. Then $A_1 \uplus A_2$ is a new context such that $x:S \in A_1 \uplus A_2$ if and only if
\[
S = \left\{ \begin{array}{ll}
               S_1 \cup S_2 & \mbox{if $x:S_1 \in A_1$ and $x:S_2 \in A_2$} \\
               S_1 & \mbox{if $x:S_1 \in A_1$ and $\neg \exists S. x:S \in A_2$} \\
               S_2 & \mbox{if $x:S_2 \in A_2$ and $\neg \exists S. x:S \in A_1$} 
            \end{array}
     \right. 
\]
\end{definition}
From now on when we write $A \uplus \{x:S\}$ we assume that $x$ does not
occur in $A$. 
We are now able to formalize the notion of term expansion.
\begin{definition}
Given a pair $M:\sigma$, where $M$ is a term and $\sigma$ an ACI-intersection type, a term $N$ and an expansion context $A$ we define a relation ${\cal E}_I(M:\sigma) \lhd (N,A)$ called {\em ACI-expansion}. If $A$ is empty we shall write just ${\cal E}_I(M:\sigma) \lhd N$. 
Expansion is defined by:
\[
\begin{array}{rll}
  {\cal E}_I(x:\tau) & \lhd & (y,\{x:\{y:\tau\}\}) \\ 
                   &   & \: \: \: \: \mbox{if $x \neq y$} \\ 
  {\cal E}_I(\lambda x.M:\tau_1 \cap \dotsb \cap \tau_n \rightarrow \sigma) & \lhd & (\lambda x_1 \ldots x_n.M^*,A) \\ 
  &   &  \: \: \: \: \mbox{if $x$ occurs in $M$ and} \\ 
  &   &  \: \: \: \: {\cal E}_I(M:\sigma) \lhd (M^*,A \cup \{x:\{x_1:\tau_1, \dots ,x_n:\tau_n\}\})  \\ 
{\cal E}_I(\lambda x.M:\tau \rightarrow \sigma) & \lhd & (\lambda y.M^*,A) \\ 
  &  & \: \: \: \: \mbox{if  $x$ does not occur in $M$,} \\ 
  &  & \: \: \: \: \mbox{$y$ is a fresh variable and} \\
  &  & \: \: \: \:  {\cal E}_I(M:\sigma) \lhd (M^*,A)  \\ 
 {\cal E}_I(MN:\sigma) & \lhd & (M_0N_1 \dots N_k, A_0 \uplus A_1 \uplus \dotsb \uplus A_n) \\ 
  &  & \: \: \: \: \mbox{if for some $k>0$ and $\tau_1, \ldots \tau_k$,} \\
  &  & \: \: \: \: {\cal E}_I(M:\tau_1 \cap \dotsb \cap \tau_k \rightarrow \sigma) \lhd (M_0,A_0) \mbox{ and} \\ 
  &  & \: \: \: \: {\cal E}_I(N:\tau_i) \lhd (N_i,A_i), (1 \leq i \leq k)
\end{array}
\]
\end{definition}
From now on if ${\cal E}_I(M:\sigma) \lhd (N,A)$ we will refer to $N$ as
an expanded version of $M$. We will sometimes omit $A$ when it is empty. 
We will now present an illustrating example.

\begin{example}
\label{ex1}
Let $I \equiv \lambda x.x$ and $M \equiv \lambda x.xx$. Let us show step by step how to calculate an expansion of $(MI:\alpha \rightarrow \alpha)$:
$${\cal E}_I(x:(\alpha \rightarrow \alpha) \rightarrow (\alpha \rightarrow \alpha)) \lhd (x_1,\{x:\{x_1:(\alpha \rightarrow \alpha) \rightarrow (\alpha \rightarrow \alpha)\}\})$$
and
$${\cal E}_I(x:\alpha \rightarrow \alpha) \lhd (x_2,\{x:\{x_2:\alpha \rightarrow \alpha\}\})$$
thus
$${\cal E}_I(xx:\alpha \rightarrow \alpha) \lhd (x_1x_2,\{x:\{x_1:(\alpha \rightarrow \alpha) \rightarrow (\alpha \rightarrow \alpha),x_2:\alpha \rightarrow \alpha\}\})$$
and
$${\cal E}_I(\lambda x.xx:(((\alpha \rightarrow \alpha) \rightarrow (\alpha \rightarrow \alpha)) \cap (\alpha \rightarrow \alpha)) \rightarrow \alpha \rightarrow \alpha) \lhd \lambda x_1 x_2.x_1 x_2$$
It easy to show that
$${\cal E}_I(I:\alpha \rightarrow \alpha) \lhd I$$
and
$${\cal E}_I(I:(\alpha \rightarrow \alpha) \rightarrow (\alpha \rightarrow \alpha)) \lhd I$$
thus
$${\cal E}_I(((\lambda x.xx)I):\alpha \rightarrow \alpha) \lhd (\lambda
x_1x_2.x_1x_2)II$$
Note that if
$${\cal E}_I(xx:\alpha \rightarrow \alpha) \lhd (x_1x_2,\{x:\{x_1:(\alpha
\rightarrow \alpha) \rightarrow (\alpha \rightarrow \alpha),x_2:\alpha
\rightarrow \alpha\}\})$$
it is also true that
$${\cal E}_I(xx:\alpha \rightarrow \alpha) \lhd (x_1x_2,\{x:\{x_2:\alpha
\rightarrow \alpha,x_1:(\alpha
\rightarrow \alpha) \rightarrow (\alpha \rightarrow \alpha)\}\})$$
because $\{x_1:(\alpha
\rightarrow \alpha) \rightarrow (\alpha \rightarrow \alpha),x_2:\alpha
\rightarrow \alpha\}$ is a set and thus there is not a fixed order
among its elements.
Thus we also have
$${\cal E}_I(\lambda x.xx:((\alpha \rightarrow \alpha) \cap ((\alpha \rightarrow \alpha) \rightarrow
(\alpha \rightarrow \alpha))))
\rightarrow \alpha \rightarrow \alpha) \lhd \lambda x_2 x_1.x_1 x_2$$
and consequently
$${\cal E}_I(((\lambda x.xx)I):\alpha \rightarrow \alpha) \lhd (\lambda
x_2x_1.x_1x_2)II$$
\end{example}
Note that the result of ACI-expansion is a term typable in the Curry Type System, not necessarily linear. For example the expansion of $\lambda f x.f(f x)$ using type $(\alpha \ra \alpha) \ra \alpha \ra \alpha$ is the term  $\lambda f_1 x_1.f_1(f_1 x_1)$.

We now show that terms that we can expand are exactly the terms
typable in an Intersection Type System i.e. the strongly normalizable terms. 
Let us first define two functions which transform expansion contexts in type environments and vice versa. 

\begin{definition}
Let $\Gamma$ be a type environment and  $\{x_1, \ldots, x_n\}$ be fresh term variables. Then $e(\Gamma)$ is the expansion context defined thus:
$$e(\Gamma) = \{x:\{x_1:\tau_1, \ldots, x_n:\tau_n\} \mid x:\tau_1 \cap \dotsb \cap \tau_n \in \Gamma \}$$
\end{definition}

\begin{definition}
Let $A$ be an expansion context. Then $l(A)$ is the type environment defined thus:
$$l(A) = \{x:\tau_1 \cap \dotsb \cap \tau_n \mid x:\{x_1:\tau_1, \ldots, x_n:\tau_n\} \in A\}$$
\end{definition}

\begin{restatable}{mylemma}{diste}
\label{dist_e}
Let $\Gamma_1$ and $\Gamma_2$ be type environments. Then
$$e(\Gamma_1) \uplus e(\Gamma_2) = e(\Gamma_1 \wedge \Gamma_2)$$
\end{restatable}

\begin{restatable}{mylemma}{distl}
\label{dist_l}
Let $A_1$ and $A_2$ be two expansion contexts. Then
$$l(A_1) \wedge l(A_2) = l(A_1 \uplus A_2)$$
\end{restatable}
We will now proceed with some auxiliary lemmas before presenting the main theorem.
\begin{restatable}{mylemma}{times}
\label{times}
Let ${\cal E}_I(M:\sigma) \lhd (N,A \uplus \{x:\{x_1:\tau_1, \ldots,
x_k:\tau_k\}\})$. Then the number of free occurrences of $x$ in $M$ is greater or equal to $k$.
\end{restatable}

\begin{restatable}{thm}{typeexpone}
\label{theorem:type_exp1}
Let $M$ be a $\lambda$-term such that there is an environment $\Gamma$ and an intersection type $\sigma$ such that $\Gamma \vdash_{\cap} M:\sigma$. Then there is a term $N$ such that ${\cal E}_I(M:\sigma) \lhd (N,e(\Gamma))$.
\end{restatable}

\begin{restatable}{mylemma}{lemmaexptwo}
\label{type_exp2}
Let $M$ be a $\lambda$-term such that there is an expansion context $A$, an intersection type $\sigma$, and a term $N$ such that ${\cal E}_I(M:\sigma) \lhd (N,A)$. Then $l(A) \vdash_\cap M:\sigma$.
\end{restatable}

\begin{restatable}{thm}{expsn}
\label{exp_sn}
Let $M$ be a $\lambda$-term. Then $M$ is strongly normalizable if and only if there are a term $N$, an expansion context $A$ and a type $\sigma$ such that ${\cal E}_I(M:\sigma) \lhd (N,A)$.
\end{restatable}

\subsubsection{ACI-Expansion and the Curry Type System}

In \cite{BucciarelliLPS99} a translation from intersection types to simple types
was given and used to show that derivations in an intersection type system with idempotent intersections can be transformed into terms typed in the Curry Type System. Here we show that our definition of ACI-expansion also preserves this translation. In fact, let ${\cal T}$ be the translation from intersection types to simple types defined in \cite{BucciarelliLPS99}. Then, if $M$ is typable in the intersection type system with type $\sigma$, and ${\cal E}_I(M:\sigma) \lhd (N,A)$ then $N$ is typable in the Curry Type System with type ${\cal T}(\sigma)$.

\begin{definition}
${\cal T}$ is a translation from intersection types to simple types defined by:
\begin{enumerate}
\item ${\cal T}(\alpha) = \alpha$, if $\alpha$ is a type variable;
\item ${\cal T}((\tau_1 \cap \dotsb \cap \tau_n) \rightarrow \sigma)$ $=$ ${\cal T}(\tau_1) \rightarrow \dotsb \rightarrow {\cal T}(\tau_n) \rightarrow {\cal T}(\sigma)$.
\end{enumerate}
\end{definition}
${\cal T}$ will be used later in the paper also for similar functions applied to linear and ordered types. Their use is clear  in each context.
The previous definition can be extended to expansion contexts:
\begin{definition}
Let ${\cal T}_e$ be a translation from expansion contexts to bases defined thus:
\begin{enumerate}
\item ${\cal T}_e(\emptyset) = \emptyset$;
\item ${\cal T}_e(A \cup \{x:\{x_1:\tau_1, \ldots, x_n:\tau_n\}\}) = {\cal T}_e(A) \cup \{x_1:{\cal T}(\tau_1), \ldots, x_n:{\cal T}(\tau_n)\}$.
\end{enumerate}
\end{definition}

\begin{restatable}{thm}{exptypes}
\label{exp_types}
Let ${\cal E}_I(M:\sigma) \lhd (N,A)$. Then ${\cal T}_e(A) \vdash_C N:{\cal T}(\sigma)$, where $\vdash_C$ stands for type derivation in the Curry Type System.

\end{restatable}

\begin{restatable}{thm}{InterCurry}

Let $M$ be a $\lambda$-term such that $\Gamma \vdash_{\cap} M:\sigma$ in the Intersection Type System. Then there is a basis $\Gamma_C$ and a term $N$ such that $\Gamma_C \vdash_C N:{\cal T}(\sigma)$, where $\vdash_C$ stands for type derivation in the Curry Type System.

\end{restatable}

This theorem has, as a corollary, that if a term $M$ is typable in the intersection type system with a simple type, then there is an expanded term with the same type derivable in the Curry type system. Just notice that ${\cal T}(\sigma) = \sigma$ when $\sigma$ is a simple type.

\subsubsection{Weak Head Reduction}

In this section we show that ACI-expansion is preserved by a notion of
reduction that is used in the implementation of functional
programming languages: weak head reduction. This guarantees that the
weak head normal form of a term $M$ has an expanded version, which is a weak head normal form of an expanded version of $M$.

We first present one Lemma that is going to be used in the study of the
preservation of expansion by reduction. 

\begin{restatable}{mylemma}{subsub}
\label{sub_sub}
Let ${\cal E}_I(M:\sigma) \lhd (M_0,A_0 \uplus \{x:\{x_1:\tau_1, \ldots,
x_k:\tau_k\}\})$ and ${\cal E}_I(N:\tau_i) \lhd (N_i,A_i)$ for $i \in
\{1,\ldots,k\}$. Then ${\cal E}_I(M[N/x]:\sigma) \lhd (M_0[N_1/x_1,\ldots,N_k/x_k],A_0 \uplus \dotsb \uplus A_k)$
\end{restatable}

Functional language compilers \cite{Jones87} consider only weak-head reduction and stop evaluation when a weak-head normal form (a constant or a $\lambda$-abstraction) is reached. Weak-head normal forms are sufficient because printable results only belong to basic domains.
The following definition of {\em weak head reduction} appears in \cite{Fradet94}:
\begin{definition}
{\em Weak head reduction} $\underset{w}{\rightarrow}$ is defined by:
$$(\lambda x.M)N \underset{w}{\rightarrow} M[N/x]$$
and
$$
\frac{{\textstyle M \underset{w}{\rightarrow} M'}}{{\textstyle MN
    \underset{w}{\rightarrow} M'N}} 
$$
We denote by $\underset{w}{\twoheadrightarrow}$ the reflexive and transitive closure of $\underset{w}{\rightarrow}$. Closed weak head normal forms are abstractions $\lambda x.M$.
\end{definition}

We first define an inclusion relation between expansion contexts as follows:
\begin{definition}
Let $A_1$ and $A_2$ be two expansion contexts. $A_1 \sqsubseteq A_2$
if and only if:
$$ x:S_1 \in A_1 \Rightarrow x:S_2 \in A_2 \mbox{ and } S_1 \subseteq S_2.$$
\end{definition}  
We will now show that ACI-Expansion preserves weak head reduction in the sense that the following diagram commutes: 


\centerline{
\xymatrix{
   M_1 \ar[r]_w \ar[d]_{{\cal E}_I}  &  M_2 \ar[d]^{{\cal E}_I}  \\
   N_1 \ar@{->>}[r]_w           &  N_2}
}

For this we need some auxiliary lemmas.

\begin{restatable}{mylemma}{sub}
Let $(\lambda x.M)N$ be a redex in the $\lambda$-calculus. Let ${\cal E}_I((\lambda x.M)N:\sigma) \lhd (N_1,A_1)$
Then there is a term $N_2$ such that ${\cal E}_I(M[N/x]:\sigma) \lhd (N_2,A_2)$, $A_2 \sqsubseteq A_1$ and $N_1 \underset{\beta}{\twoheadrightarrow} N_2$.
\label{sub}
\end{restatable}

\begin{restatable}{thm}{tfive}
\label{tfive}
Let ${\cal E}_I(M_1:\sigma) \lhd (N_1,A_1)$ and $M_1
\underset{w}{\rightarrow} M_2$. Then there is a term $N_2$ such that
${\cal E}_I(M_2:\sigma) \lhd (N_2,A_2)$, $N_1 \underset{w}{\twoheadrightarrow}
N_2$ and $A_2 \sqsubseteq A_1$.
\end{restatable}

\begin{definition}
Let $t$ and $u$ be $w$-reductions starting, respectively, by $M_0$ and $N_0$:
$$t:M_0 \underset{w}{\rightarrow} M_1 \underset{w}{\rightarrow} M_2 \underset{w}{\rightarrow} \cdots$$
$$u:N_0 \underset{w}{\twoheadrightarrow} N_1 \underset{w}{\twoheadrightarrow} N_2 \underset{w}{\twoheadrightarrow} \cdots$$
We say that $u$ is an {\em expansion} of $t$ if there are expansion contexts $A_0, \ldots, A_k$ and a type $\sigma$ such that:
\begin{enumerate}
\item $A_0 \sqsupseteq A_1 \sqsupseteq A_2 \sqsupseteq \cdots$,
\item ${\cal E}_I(M_i:\sigma) \lhd (N_i,A_i)$ for $i \geq 0$.
\end{enumerate}
\end{definition}

The following corollary of Theorem \ref{tfive} makes explicit the simple fact that every finite $w$-reduction can be expanded. It holds trivially by successive applications of Theorem \ref{tfive} to every $w$-reduction step in $t$.

\begin{corollary}[of Theorem \ref{tfive}]
Every finite $w$-reduction $t$, can be expanded to another $w$-reduction (not necessarily unique).
\end{corollary}

We saw that expansion is preserved by weak head reduction. This does
not happen with $\beta$-reduction. 
In fact we may have $M_1
\underset{\beta}{\rightarrow} M_2$, ${\cal E}_I(M_1:\sigma) \lhd
(N_1,A_1)$ and there is not a type $\tau$ such that ${\cal
  E}_I(M_2:\tau) \lhd (N_2,A_2)$ and $N_1
\underset{\beta}{\twoheadrightarrow} N_2$. Note that there is an expanded
version, $P$, of $M_2$ (because $M_1$ is strongly normalizable thus
$M_2$ is also strongly normalizable and thus, by Theorem \ref{exp_sn}, it
has an expanded version). The point here is that $N_1 \not
\underset{\beta}{\twoheadrightarrow} P$ for no expanded version $P$ of $M_2$.
To see this let $M_1 \equiv \lambda x.(\lambda y.z)xx$ and $M_2 \equiv \lambda x.zx$.
We have:
$$\lambda x.(\lambda y.z)xx \underset{\beta}{\rightarrow}  
\lambda x.zx$$
$${\cal E}(\lambda x.(\lambda y.z)xx:\alpha_1 \cap
  \alpha_2 \rightarrow \beta) \lhd (\lambda x_1x_2.(\lambda
  y_1.z_1)x_1x_2,\{z:\{z_1:\alpha_2 \rightarrow \beta\}\})$$
and
$$ \lambda x_1x_2.(\lambda y_1.z_1)x_1x_2
\underset{\beta}{\rightarrow} \lambda x_1x_2.z_1x_2$$
Now note that, as $x$ occurs in $zx$ once, it follows from Lemma
\ref{times} that any expansion of $\lambda x.zx$ is of the form
$\lambda x_1.M$ where $M$ is one expansion of $zx$. Thus $\lambda
x_1x_2.z_1x_2$ cannot be an expansion of $\lambda x.zx$ for any
type. 
If preservation of expansion by $\beta$-reduction is not viewed as a
goal by itself, then the lack of
this property is not a problem, because it holds for a notion of
reduction that is used for functional programming languages. 

\subsubsection{ACI-Expansion and the Relevant Type System}

Here we will study ACI-expansion applied only to $\lambda$I-terms. Note that it is the same relation, ${\cal E}_I$, defined in the previous section, but we now restrict its domain to the set of $\lambda$I-terms. Thus the same symbol ${\cal E}_I$ will be used, overloaded, in this section to evoke this analogy.

We recall that the $\lambda I$-calculus is a restriction of the $\lambda$-calculus where in terms of the form $\lambda x.M$, $x$ occurs free in $M$.

We now show that terms in the range of ${\cal E}_I$, when its domain is the $\lambda I$-calculus, are typed in the Relevant Type System.

\begin{restatable}{thm}{exptypesRelevant}
\label{exp_types_Relevant}
Let $M$ be a $\lambda I$-term such that ${\cal E}_I(M:\sigma) \lhd (N,A)$. Then ${\cal T}_e(A) \vdash_R N:{\cal T}(\sigma)$, where $\vdash_R$ stands for type derivation in the Relevant Type System.
\end{restatable}

\begin{restatable}{thm}{Relevant}
Let $M$ be a $\lambda I$-term such that $\Gamma \vdash_{\cap} M:\sigma$ in the Intersection Type System. Then there is a basis $\Gamma_R$ and a term $N$ such that $\Gamma_R \vdash_R N:{\cal T}(\sigma)$, where $\vdash_R$ stands for type derivation in the Relevant Type System.

\end{restatable}

This theorem has, as a corollary, that if a $\lambda I$-term $M$ is typable in the Intersection Type System with a simple type, then there is an expanded term with the same type derivable in the Relevant Type System. Just notice that ${\cal T}(\sigma) = \sigma$ when $\sigma$ is a Curry type.

\subsubsection{Reduction}

We now show that $\beta$-reduction is preserved by ACI-expansion for the $\lambda I$-calculus, where erasing is not allowed. This means that for the $\lambda I$-calculus the following diagram commutes:

\centerline{
\xymatrix{
   M_1 \ar[r]_\beta \ar[d]_{{\cal E}_I}  &  M_2 \ar[d]^{{\cal E}_I}  \\
   N_1 \ar@{->>}[r]_\beta           &  N_2}
}

\begin{restatable}{mylemma}{SubI}
Let $(\lambda x.M)N$ be a redex in the $\lambda I$-calculus. Let
${\cal E}_I((\lambda x.M)N:\sigma) \lhd (N_1,A)$.
Then there is a term $N_2$ such that
${\cal E}_I(M[N/x]:\sigma) \lhd (N_2,A)$
and $N_1 \underset{\beta}{\twoheadrightarrow} N_2$.
\label{subI}
\end{restatable}

\begin{restatable}{thm}{red}
\label{red}
Let $M_1$ and $M_2$ be two terms in the $\lambda I$-calculus. Let
${\cal E}_I(M_1:\sigma) \lhd (N_1,A)$ and $M_1
\underset{\beta}{\rightarrow} M_2$. Then there is a term $N_2$ such
that ${\cal E}(M_2:\sigma) \lhd (N_2,A)$ and $N_1
 \underset{\beta}{\twoheadrightarrow} N_2$. 
\end{restatable}

\subsection{From Intersection Types to Linear Types: AC-Expansion}

AC-Expansion was defined in \cite{FloridoD04} to linearize the strongly normalizable terms. AC-expansion relies on the use of non idempotent intersection types, which give us a one-to-one relation between the number of types in an intersection and the number of occurrences of a formal parameter $x$ in a function $\lambda x.M$. This means that expanded terms will be affine or linear terms (depending on the range of expansion). 
Thus the definition of AC-Expansion (which can be found in \cite{FloridoD04}) is similar to ACI-Expansion, using non idempotent intersections and the following different expansion rule for variables:

\[
\begin{array}{rll}
  {\cal E}_C(x:\tau) & \lhd & (y,\{x:\{y:\tau\}\}) \\ 
                   &   & \: \: \: \: \mbox{if $x$ is a variable and $y$ is a fresh variable}
\end{array}  
\]
Proofs of theorems in this subsection can be found in \cite{FloridoD04}.
\begin{example}
Let $I \equiv \lambda x.x$ and $M \equiv (\lambda f.f(\lambda x.xx)(fI))I$
Then
$${\cal E}_C(M:\alpha \rightarrow \alpha) \lhd (\lambda f_1f_2f_3.f_1(\lambda x_1x_2.x_1x_2)(f_2I)(f_3I))III$$
\end{example}
Notice in this example the use of type information to control the number of expansions. In $$(\lambda f.f(\lambda x.xx)(fI))I$$ the fact that $f$ is going to be the identity function gives $f$ three different types, one for the identity function applied to $\lambda x.xx$, and two more types, one for each type in the intersection in the argument type of $f(\lambda x.xx)$. These three types give rise to the three new expansion variables $f_1$, $f_2$ and $f_3$. The intersection of two types in the argument type of $f(\lambda x.xx)$ gives rise to the two new terms $(f_2I)$ and $(f_3I)$.

From now on, to stress that expanded versions are affine or linear, when we have ${\cal E}_C(M:\sigma) \lhd (N,A)$ we will refer to $N$ as
one linear version of $M$.

\subsubsection{AC-Expansion and the Affine Type System}

\begin{definition}
${\cal T}$ is a translation from intersection types to linear types defined by:
\begin{enumerate}
\item ${\cal T}(\alpha) = \alpha$, if $\alpha$ is a type variable;
\item ${\cal T}((\tau_1 \cap \dotsb \cap \tau_n) \rightarrow \sigma)$ $=$ ${\cal T}(\tau_1) \llto \dotsb \llto {\cal T}(\tau_n) \llto {\cal T}(\sigma)$.
\end{enumerate}
\end{definition}

\begin{theorem}
\label{exp_typesAffine}
Let ${\cal E}_C(M:\sigma) \lhd (N,A)$. Then ${\cal T}_e(A) \vdash_A N:{\cal T}(\sigma)$, where $\vdash_A$ stands for type derivation in the Affine Type System.

\end{theorem}

\begin{theorem}

Let $M$ be a $\lambda$-term such that $\Gamma \vdash_{cap} M:\sigma$ in the Intersection Type System. Then there is a basis $\Gamma_A$ and a term $N$ such that $\Gamma_A \vdash_A N:{\cal T}(\sigma)$, where $\vdash_A$ stands for type derivation in the Affine Type System.

\end{theorem}

\subsubsection{Weak-Head Reduction}

AC-expansion is also preserved by  weak head reduction. This guarantees that the
weak head normal form of a term $M$ has an expanded version, which is a weak
head normal form of an expanded version of $M$.

We first present one lemma that is going to be used in the study of the
preservation of expansion by reduction. 

\begin{mylemma}
\label{sub_sub}
Let ${\cal E}_C(M:\sigma) \lhd (M_0,A_0 \uplus \{x:\{x_1:\tau_1, \ldots,
x_k:\tau_k\}\})$ and ${\cal E}_C(N:\tau_i) \lhd (N_i,A_i)$ for $i \in
\{1,\ldots,k\}$. Then
$${\cal E}_C(M[N/x]:\sigma) \lhd (M_0[N_1/x_1,\ldots,N_k/x_k],A_0 \uplus
\dotsb \uplus A_k)$$
\end{mylemma}

AC-Expansion preserves weak head reduction, thus the following diagram commutes: 


\centerline{
\xymatrix{
   M_1 \ar[r]_w \ar[d]_{{\cal E}_C}  &  M_2 \ar[d]^{{\cal E}_C}  \\
   N_1 \ar@{->>}[r]_w           &  N_2}
}

\begin{theorem}
\label{t5}
Let ${\cal E}_C(M_1:\sigma) \lhd (N_1,A_1)$ and $M_1
\underset{w}{\rightarrow} M_2$. Then there is an affine term $N_2$ such that
${\cal E}_C(M_2:\sigma) \lhd (N_2,A_2)$, $N_1 \underset{w}{\twoheadrightarrow}
N_2$ and $A_2 \sqsubseteq A_1$.
\end{theorem}

We saw that AC-expansion was preserved by weak head reduction. The same example used to prove that $\beta$-reduction is not preserved by ACI-expansion holds to show the same property for AC-expansion. 
In fact we may have $M_1
\underset{\beta}{\rightarrow} M_2$, ${\cal E}_C(M_1:\sigma) \lhd
(N_1,A_1)$ and there is not a type $\tau$ such that ${\cal
  E}_C(M_2:\tau) \lhd (N_2,A_2)$ and $N_1
\underset{\beta}{\twoheadrightarrow} N_2$. Let $M_1 \equiv \lambda x.(\lambda y.z)xx$ and $M_2 \equiv \lambda x.zx$.
We have:
$$\lambda x.(\lambda y.z)xx \underset{\beta}{\rightarrow}  
\lambda x.zx$$
$${\cal E}_C(\lambda x.(\lambda y.z)xx:\alpha_1 \cap
  \alpha_2 \rightarrow \beta) \lhd (\lambda x_1x_2.(\lambda
  y_1.z_1)x_1x_2,\{z:\{z_1:\alpha_2 \rightarrow \beta\}\})$$
and
$$ \lambda x_1x_2.(\lambda y_1.z_1)x_1x_2
\underset{\beta}{\rightarrow} \lambda x_1x_2.z_1x_2$$
Now note that, as $x$ occurs in $zx$ once,  $\lambda
x_1x_2.z_1x_2$ cannot be an expansion of $\lambda x.zx$ for any type.  

\subsubsection{AC-Expansion and the Linear Type System}

Here we will study AC-expansion applied only to $\lambda$I terms. Note that it is the same relation, ${\cal E}_C$, defined in the previous section, but we now restrict its domain to the set of $\lambda$I-terms. Thus the same symbol ${\cal E}_C$ will be used overloaded.

We show that terms in the range of ${\cal E}_C$ when its domain is the $\lambda I$-calculus are typed in the Linear Type System.

\begin{restatable}{thm}{exptypesLinear}
\label{exp_types_Linear}
Let $M$ be a $\lambda I$-term such that ${\cal E}_C(M:\sigma) \lhd (N,A)$. Then ${\cal T}_e(A) \vdash_L N:{\cal T}(\sigma)$, where $\vdash_L$ stands for type derivation in the Linear Type System.

\end{restatable}

\begin{restatable}{thm}{IL}

Let $M$ be a $\lambda I$-term such that $\Gamma \vdash_{\cap} M:\sigma$ in the Intersection Type System. Then there is a basis $\Gamma_L$ and a term $N$ such that $\Gamma_L \vdash_L N:{\cal T}(\sigma)$, where $\vdash_L$ stands for type derivation in the Linear Type System.

\end{restatable}

\subsubsection{Reduction}

We show that $\beta$-reduction is preserved by AC-expansion for the $\lambda I$-calculus, where erasing is not allowed. This means that for the $\lambda I$-calculus the following diagram commutes:

\centerline{
\xymatrix{
   M_1 \ar[r]_\beta \ar[d]_{{\cal E}_C}  &  M_2 \ar[d]^{{\cal E}_C}  \\
   N_1 \ar@{->>}[r]_\beta           &  N_2}
}
\begin{restatable}{mylemma}{SubI}
Let $(\lambda x.M)N$ be a redex in the $\lambda I$-calculus. Let
${\cal E}_C((\lambda x.M)N:\sigma) \lhd (N_1,A)$
Then there is a linear term $N_2$ such that
${\cal E}_C(M[N/x]:\sigma) \lhd (N_2,A)$
and $N_1 \underset{\beta}{\twoheadrightarrow} N_2$.
\label{subI}
\end{restatable}

\begin{restatable}{thm}{LinearI}
Let $M_1$ and $M_2$ be two terms in the $\lambda I$-calculus. Let
${\cal E}_C(M_1:\sigma) \lhd (N_1,A)$ and $M_1
\underset{\beta}{\rightarrow} M_2$. Then there is a term $N_2$ such
that ${\cal E}_C(M_2:\sigma) \lhd (N_2,A)$ and $N_1
 \underset{\beta}{\twoheadrightarrow} N_2$. 
\end{restatable}

\subsection{From Intersection Types to Ordered Types: Ordered Expansion}

Here we define the new notion of {\em ordered expansion}, which relates terms typable by non-idempotent and non-commutative intersections with terms typable in the ordered type system. As order now matters, expansion contexts will be defined as lists. 
\begin{definition}
${\cal T}$ is a translation from intersection types to ordered types defined by:
\begin{enumerate}
\item ${\cal T}(\alpha) = \alpha$, if $\alpha$ is a type variable;
\item ${\cal T}((\tau_1 \cap \dotsb \cap \tau_n) \rightarrow \sigma)$ $=$ ${\cal T}(\tau_1) \llto_r \dotsb \llto_r {\cal T}(\tau_n) \llto_r {\cal T}(\sigma)$.
\end{enumerate}
\end{definition}
\begin{definition}
A {\em variable expansion} is an expression of the form
$x:S$
where $x$ is a variable and $S$ is a list of pairs of the form $y:\tau$ where $y$ is a variable and $\tau$ an intersection type ($x:S$ should be read informally as ``$x$ expands to the variables in $S$'').
\end{definition}

\begin{definition}
An {\em expansion context} $A$ is a finite list of variable expansions,
$A = [x_1:S_1, \ldots ,x_n:S_n]$,
where the variables $\{x_1, \ldots, x_n\}$ are all different and the $S_i$ have no elements in common.
\end{definition}
We now define an operation that appends two expansion contexts.
\begin{definition}
Let $A_1$ and $A_2$ be two expansion contexts. Then $A_1 + A_2$ is a new expansion context define inductively as:
\[
A_1+A_2 = \left\{ \begin{array}{ll}
               A_1 & \mbox{if $A_2 = \octx{\,}$} \\
               (A_1',x:S_1,S_2,A_1'')+A_2' & \mbox{if $A_1= A_1',x:S_1,A_1''$ and $A_2=x:S_2,A_2'$} \\
               (A_1,x:S_2)+A_2' & \mbox{otherwise} 
            \end{array}
     \right. 
\]
\end{definition}
From now on when we write $A + [x:S]$ we assume that $x$ does not
occur in $A$. 
We are now able to formalize the notion of ordered expansion:

\begin{definition}[Ordered Expansion]
The \deft{ordered expansion} relation $\expo(M:\sig) \goto (N^{\tau},A)$ for $M,N$ (pure) $\l$-terms, $\sig$ an intersection type and $A$ an expansion context is inductively defined by:
\begin{eqnarray*}
\expo(x:\sig) &\goto& (y^{\at{\sig}},\octx{x:\octx{y:\at{\sig}}}), y \mbox{ fresh}\\
\expo(\l{x}.M : \sig_1 \cap \dotsb \cap \sig_n \llto \sig)  &\goto& (\l{y_1 \dots y_n}.M_0^{\at{\sig_1} \llto_{r} \dotsb \llto_{r} \at{\sig_n} \llto_r \at{\sigma}}, A), \\ 
&& \mbox{ if $x \in \fv{M}$ and }\\
                    &  & \expo(M : \sig) \goto (M_0^{\at{\sigma}}, A+\octx{x : \octx{x_1:\at{\sig_1},\dots,x_n:\at{\sig_n}}}) \\
\expo(\l{x}.M : \sig_1 \cap \dotsb \cap \sig_n \llto \sig) & \goto & (\l{x_1 \dots x_n}.M_0^{\at{\sig_1} \llto_{l} \dotsb \llto_{l} \at{\sig_n} \llto_l \at{\sigma}},A), \\ && \mbox{ if $x \in \fv{M}$ and }\\
                &  & \expo(M : \sig) \goto (M_0^{\at{\sigma}},  \octx{x : \octx{x_n:\at{\sig_n},\dots,x_1:\at{\sig_1}}} + A) \\
\expo(M N : \sig) & \goto & ((M_0 N_1 \dots N_m)^{\at{\sig}},  A_0+ A_1+ \dotsb+ A_m), \\
                & &\mbox{ if for some $m > 0$ and $\sig_1,\dots,\sig_m$} \\
                  &  & \expo(M : \sig_1 \cap \dots \cap \sig_m \llto \sig) \goto (M_0^{\at{\sig_1} \llto_r \dotsb \llto_r \at{\sig_m} \llto_r \at{\sig}}, A_0) \\ && \mbox{ and } \left( \expo(N : \sig_i) \goto (N_i^{\at{\sig_i}}, A_i) \right)_{\iom}\\
\expo(M N : \sig) & \goto & ((M_0 N_1 \dots N_m)^{\at{\sig}}, A_m +\dotsb +A_1+A_0),\\
                & &\mbox{ if for some $m > 0$ and $\sig_1,\dots,\sig_m$} \\
                  &  & \expo(M : \sig_1 \cap \dots \cap \sig_m \llto \sig) \goto (M_0^{\at{\sig_1} \llto_l \dotsb \llto_l \at{\sig_m} \llto_l \at{\sig}}, A_0) \\ &&\mbox{ and } \left( \expo(N : \sig_i) \goto (N_i^{\at{\sig_i}}, A_i) \right)_{\iom}                  
\end{eqnarray*}
\label{OExpansion}
\end{definition}


\subsubsection{Ordered Expansion and Ordered Types}

\begin{definition}
Let $\mathcal{T}_e$ be a translation from expansion contexts to bases defined thus:
\begin{enumerate}
\item $\cetct{\epsilon} = \epsilon$;
\item $\cetct{A + [x:[x_1:\tau_1, \ldots, x_n:\tau_n]]} = {\cal T}_e(A),x_1:\tau_1, \ldots, x_n:\tau_n$.
\end{enumerate}
\end{definition}

\begin{restatable}{thm}{orderexp} Let $M$ be a $\l_I$-term. If $\expo(M:\sigma) \goto (N^{\at{\sigma}}, A)$, then $\cetct{A} \vdash_o N:\at{\sigma}$.
\label{orderexp}
\end{restatable}
We will now present an example illustrating Definition \ref{OExpansion} and Theorem \ref{orderexp}.
\begin{example}Let $M\equiv (\l x.xz)z$. The ordered expansion of $M$ is calculated step by step as:
\[
\begin{array}{l}
\expo((\l x.xz)z:\beta) = (((\l x_1.x_1z_1)z_2)^{\beta},\octx{z:\octx{z_2:\alpha \rew_r \beta,z_1:\alpha}})\\
\qquad \expo(\l x.xz:(\alpha \rew \beta)\rew \beta) = ((\l x_1.x_1z_1)^{(\alpha \rew_r \beta)\rew_l \beta},\octx{z:\octx{z_1:\alpha}})\\
\qquad \qquad \expo(xz:\beta) = ((x_1z_1)^{\beta},\octx{x:\octx{x_1:\alpha\rew_r \beta],z:[z_1:\alpha}})\\
\qquad \qquad \qquad \expo(x:\alpha\rew_r \beta) = (x_1^{\alpha\rew_r \beta},\octx{x:\octx{x_1:\alpha\rew_r \beta}})\\
\qquad \qquad \qquad \expo(z:\beta) = (z_1^{\beta},\octx{z:\octx{z_1:\beta}})\\
\qquad \expo(z: \alpha \rew \beta) = (z_2^{\alpha \rew_r\beta},\octx{z:\octx{z_2:\alpha\rew_r \beta}})\\
\end{array}
\]
Theorem \ref{orderexp} guarantees that the expanded version of $M$ is typable in the ordered type system (in this case with the same type). The corresponding type derivation follows:
\[
\inferrule{\inferrule{\inferrule{\inferrule{ }{\octx{x_1:\alpha\rew_r \beta}\vdash_o x_1 :{\alpha \rew_r\beta}}\quad\inferrule{ }{\octx{z_1:\alpha}\vdash_o z_1 :{\beta}}}{\octx{x_1:\alpha\rew_r \beta,z_1:\alpha}\vdash_o x_1z_1 :{\beta}} }{ \octx{z_1:\alpha} \vdash_o(\l x_1.x_1z_1):(\alpha \rew_r \beta)\rew_l \beta}
\quad\inferrule{ }{\octx{z_2:\alpha \rew_r \beta} \vdash_o z_2:\alpha \rew_r \beta }}
{ \octx{z_2:\alpha \rew_r \beta,z_1:\alpha}\vdash_o (\l x_1.x_1z_1)z_2):\beta}
\]
\end{example}
\subsection{Reduction}
As it happens with ACI-expansion and AC-expansion, ordered expansion is also preserved by $\beta$-reduction for $\lambda I$-terms.
\begin{restatable}{mylemma}{occur}
Let $\expo(M,\sig) \goto (N,A_1+\octx{x:\octx{x_1:\tau_1,\dots,x_n:\tau_n}}+A_2)$, then there exist $n$ occurrences of $x$ in $M$.
\end{restatable}
\begin{restatable}{mylemma}{subst}
\label{lem:subst}
Let $\expo(M,\sig) \goto (M_0^{\at{\sig}},A_0+\octx{x:\octx{x_1:\at{\tau_1},\dots,x_n:\at{\tau_n}}}+A_{n+1})$ and $$\left(\expo(N,\tau_i)\goto(N_i^{\at{\tau_i}},A_i)\right)_{\ion}$$ then $\expo(M[N/x],\sigma) \goto ((M_0[N_1/x_1,\dots,N_n/x_n])^{\at{\sig}},A_0+_{\ion} A_i + A_{n+1}).$
\end{restatable}

\begin{restatable}{mylemma}{reduction}
\label{lem:reduction}
Let $(\lambda x.M)N$ be a redex in the $\lambda I$-calculus. Let
$\expo((\lambda x.M)N:\sigma) \goto (N_1^{\tau_1},A)$.
Then there is a linear term $N_2$ and a type $\tau_2$ such that
$\expo(M[N/x]:\sigma) \lhd (N_2^{\tau_2},A)$
and $N_1 \underset{\beta}{\twoheadrightarrow} N_2$.
\end{restatable}
\begin{restatable}{thm}{icalculus}
Let $M_1$ and $M_2$ be two terms in the $\lambda I$-calculus. Let
$\expo(M_1:\sigma) \goto (N_1^{\tau_1},A)$ and $M_1
\underset{\beta}{\rightarrow} M_2$. Then there is a term $N_2$ and a type $\tau_2$ such
that $\expo(M_2:\sigma) \goto (N_2^{\tau_2},A)$ and $N_1
 \underset{\beta}{\twoheadrightarrow} N_2$. 
\end{restatable}

\section{Conclusions}

In this paper we show that there is a strong, and somehow unexpected, relation between intersection types and the substructural type systems: idempotent intersection is related with Curry and Relevant types, commutative intersection with Linear and Affine types and associative intersection with Ordered types. This highlights a clear relation between algebraic properties of intersection types and the substructural rules: idempotent intersection is related with the contraction rule and commutative intersection with the exchange rule.
The following table relates the algebraic properties of the intersection operator used in expansion with the different type systems obtained.
\newline

\begin{center}
\begin{tabular}{|c|c|c|}
\hline
{\bf $\cap$} & {\bf Target language} & {\bf Preserves reductions} \\
\hline
\hline
ACI & Simple Types & Weak Head Reduction \\
\hline
ACI & Relevant Types & $\beta$-reduction \\
\hline
AC & Affine Types & Weak Head Reduction \\
\hline
AC & Linear Types & $\beta$-reduction \\
\hline
A & Ordered Types & $\beta$-reduction \\
\hline
\end{tabular}
\end{center}

%% file: appendix.tex
\newpage
\appendix

\section{Proofs}
\label{section:proofs}

\lemaa*

\begin{proof}
By induction on the length of the type derivation.
\end{proof}

\theorema*

\begin{proof}
By induction on the length of the type derivation and using Lemma \ref{FV}.
\end{proof}

\FVA*

\begin{proof}
Straightforward induction on the length of the typing derivation.
\end{proof}

\Affine*

\begin{proof}
The proof follows by induction on the length of the typing derivation (case $\Rightarrow$) and on structural induction on the term $M$ (case $\Leftarrow$).
\end{proof}

\FVL*

\begin{proof}
Straightforward induction on the length of the typing derivation.
\end{proof}

\LinearTerms*

\begin{proof}
The proof follows by induction on the length of the typing derivation (case $\Rightarrow$) and on structural induction on the term $M$ (case $\Leftarrow$).
\end{proof}

\sn*

\begin{proof}The {\em if} part is proved by transforming a derivation in our type system in a derivation in the Coppo-Dezani type system (which types all strongly normalizable terms). This can be done by induction in the length of the derivation tree. The {\em only-if} part is similar to the proof of the same property for the Coppo-Dezani type system presented in \cite{Amadio:1998}. The tecnhique used is to show that if a term $M[N/x]$ is typable in an intersection type system with type $\tau$ then the redex $(\lambda x.M)N$ is also typable with the same type. The result follows by lifting this property to arbitrary terms using induction on the size of the term and on the maximal length of derivations starting in the term. 
\end{proof}

\times*

\begin{proof} By  structural induction on $M$.
\end{proof}

\diste*

\begin{proof} By the definitions of $\wedge$ and $\uplus$. \end{proof}

\distl*

\begin{proof} By the definitions of $\wedge$ and $\uplus$. \end{proof}

\expsn*

\begin{proof} By Theorem \ref{sn} the stronlgy normalizable terms
are the terms typable in the intersection type system presented. The
result follows by Theorem \ref{theorem:type_exp1} and Lemma \ref{type_exp2}. \end{proof}

\InterCurry*

\begin{proof} By Theorem \ref{sn} the stronlgy normalizable terms
are the terms typable in the intersection type system presented. The
result follows by Lemmas \ref{theorem:type_exp1} and \ref{type_exp2}. \end{proof}

\exptypesRelevant*

\begin{proof} Similar to the proof of Theorem \ref{exp_types} considering that it applies to every case except to applications of the Weakening rule. 
\end{proof}

\Relevant*

\begin{proof} By Lemma \ref{theorem:type_exp1}, $\Gamma \vdash M:\sigma$ $\Rightarrow$ ${\cal E}_I(M:\sigma) \lhd (N,e(A))$. The result follows by Theorem \ref{exp_types_Relevant}. \end{proof}

\SubI*

\begin{proof} The proof is identical to the case in the proof of
Lemma \ref{sub} where $x$ occurs free in $M$. \end{proof}

\exptypesLinear*

\begin{proof} Similar to the proof of Theorem \ref{exp_typesAffine} considering that it applies to every case except to applications of the Weakening rule. 
\end{proof}

\IL*

\begin{proof} Note that  $\Gamma \vdash M:\sigma$ $\Rightarrow$ ${\cal E}_C(M:\sigma) \lhd (N,e(A))$. The result follows by Theorem \ref{exp_types_Linear}. \end{proof}

\SubI*

\begin{proof} The proof is identical to the case in the proof of
Lemma \ref{sub} where $x$ occurs free in $M$. \end{proof}

\LinearI*

\begin{proof} Similar to the proof of Theorem \ref{red} for ACI-reduction.
\end{proof}

\typeexpone*

\begin{proof} By structural induction on $M$. 
\begin{enumerate}
\item Base case: $M$ is a term-variable $x$. $\{x:\tau\} \vdash_{\cap} x:\tau$. Then ${\cal E}_I(x:\tau) \lhd (y,\{x:\{y:\tau\}\})$ where $y$ is a fresh variable. The result follows noticing that $\{x:\{y:\tau\}\} = e(\{x:\tau\})$.
\item Induction step:
\begin{enumerate}
\item $M$ is of the form $\lambda x.N$ and $x$ occurs in $N$. In this case $\Gamma \cup \{x:\tau_1 \cap \dotsb \cap \tau_n\} \vdash_{\cap} N:\sigma$. By the induction hypothesis, 
$${\cal E}_I(N:\sigma) \lhd (N',e(\Gamma) \cup \{x:\{x_1:\tau_1, \ldots, x_n:\tau_n\}\})$$
Thus by the definition of expansion
$$ {\cal E}_I(\lambda x.N: \tau_1 \cap \ldots \cap \tau_n \rightarrow \sigma)\lhd(\lambda x_1 \ldots x_n.N',e(\Gamma))$$
\item $M$ is of the form $\lambda x.N$ and $x$ does not occur in $N$. In this case $\Gamma \vdash_{\cap} N:\sigma$. By the induction hypothesis,
$${\cal E}_I(N:\sigma) \lhd (N',e(\Gamma))$$
and by the definition of expansion
$${\cal E_I}(\lambda x.N:\tau \rightarrow \sigma) \lhd (\lambda y.N',e(\Gamma))$$
\item $M$ is of the form $M_1M_2$. In this case we have 
$$ \Gamma_0 \wedge \Gamma_1 \wedge \dotsb \wedge \Gamma_n \vdash_{\cap} M_1M_2:\sigma $$
Thus
\begin{enumerate}
\item $\Gamma_0 \vdash_{\cap} M_1:\tau_1 \cap \dotsb \cap \tau_n \rightarrow \sigma$
\item $\Gamma_i \vdash_{\cap} M_2:\tau_i (1 \leq i \leq n)$
\end{enumerate}
By the induction hypothesis:
\begin{enumerate}
\item ${\cal E}_I(M_1:\tau_1 \cap \ldots \cap \tau_n \rightarrow \sigma) \lhd (M_0,e(\Gamma_0))$
\item ${\cal E}_I(M_2:\tau_i) \lhd (N_i,e(\Gamma_i)), (1 \leq i \leq n)$
\end{enumerate}
Thus, by the definition of expansion,
$$ {\cal E}_I(M_1M_2:\sigma) \lhd (M_0N_1 \ldots N_n,e(\Gamma_0) \uplus \dotsb \uplus e(\Gamma_n))$$
and finally by Lemma \ref{dist_e}
$$ {\cal E}_I(M_1M_2:\sigma) \lhd (M_0N_1 \ldots N_n,e(\Gamma_0 \wedge \dotsb \wedge \Gamma_n))$$ 
\end{enumerate}
\end{enumerate}
\end{proof}

\lemmaexptwo*

\begin{proof} By structural induction on $M$.
\begin{enumerate}
\item Base case: $M$ is a term variable $x$. In this case ${\cal E}_I(x:\tau) \lhd (y,\{x:\{y:\tau\}\})$. We have $l(\{x:\{y:\tau\}\}) = \{x:\tau\}$. Finally $\{x:\tau\} \vdash_{\cap} x:\tau$.
\item Induction step:
\begin{enumerate}
\item $M$ is of the form $\lambda x.N$, and $x$ occurs in $N$. 
$${\cal E}_I(\lambda x.N:(\tau_1 \cap \dotsb \cap \tau_n \rightarrow \sigma)) \lhd (\lambda x_1 \ldots x_n.N',A)$$ where
$$ {\cal E}_I(N:\sigma) \lhd (N',A \cup \{x:\{x_1:\tau_1, \ldots, x_n:\tau_n\}\}) $$
By the induction hypothesis
$$ l(A) \cup \{x:\tau_1 \cap \dotsb \cap \tau_n\} \vdash_{\cap} N:\sigma $$
Thus
$$ l(A) \vdash \lambda x.N:\tau_1 \cap \dotsb \cap \tau_n \rightarrow \sigma$$
\item $M$ is of the form $\lambda x.N$ and $x$ does not occur in $N$.
$${\cal E}_I(\lambda x.N:\tau \rightarrow \sigma) \lhd (\lambda y.N',A)$$ where
$${\cal E}_I(N:\sigma) \lhd (N',A)$$
By the induction hypothesis
$$l(A) \vdash_{\cap} N:\sigma$$
Thus
$$l(A) \vdash \lambda x.N : \tau \rightarrow \sigma$$
\item $M$ is of the form $M_1M_2$. In this case we have:
$${\cal E}_I((M_1M_2):\sigma) \lhd (M_0N_1 \ldots N_k,A_0 \uplus \dotsb \uplus A_k)$$
and
\begin{enumerate}
\item ${\cal E}_I(M_1:\tau_1 \cap \dotsb \cap \tau_n \rightarrow \sigma) \lhd (M_0,A_0)$
\item ${\cal E}_I(M_2:\tau_i) \lhd (N_i,A_i) (1 \leq i \leq n)$
\end{enumerate}
By the induction hypothesis
\begin{enumerate}
\item $l(A_0) \vdash_{\cap} M_1:\tau_1 \cap \dotsb \cap \tau_n \rightarrow \sigma$
\item $l(A_i) \vdash_{\cap} M_2:\tau_i (1 \leq i \leq n)$
\end{enumerate}
Thus
$$l(A_0) \wedge \dotsb \wedge l(A_n) \vdash_{\cap} M_1M_2:\sigma$$
Finally, by Lemma \ref{dist_l},
$$l(A_0 \uplus \dotsb \uplus A_n) \vdash_{\cap} M_1M_2:\sigma$$
\end{enumerate}
\end{enumerate} 
\end{proof}

\exptypes*

\begin{proof} By structural induction on $M$. 
\begin{enumerate}
\item Base case. $M$ is a term variable $x$. In this case ${\cal E}_I(x:\tau) \lhd (y,\{x:\{y:\tau\}\})$.
${\cal T}_e(\{x:\{y:\tau\}\}) = \{y:{\cal T}(\tau)\}$. The result follows by the VAR rule for the Curry type system.
\item Induction step: 
\begin{enumerate}
\item $M$ is of the form $\lambda x.N$. Suppose that $x$ occurs free in $M$. Then
$${\cal E}_I(\lambda x.N:\tau_1 \cap \dotsb \cap \tau_n \rightarrow \sigma)  \lhd  (\lambda x_1 \ldots x_n.N^*,A)$$ 
and thus
$${\cal E}_I(N:\sigma) \lhd (N^*,A \cup \{x:\{x_1:\tau_1, \ldots, x_n:\tau_n\}\})$$
By the induction hypothesis and the definition of ${\cal T}_e$:
$${\cal T}_e(A) \cup \{x_1:{\cal T}(\tau_1), \ldots ,x_n:{\cal T}(\tau_n)\} \vdash_C N^*:{\cal T}(\sigma)$$
Thus, by successive applications of the ABS-I rule:
$$ {\cal T}_e(A) \vdash_C \lambda x_1 \ldots x_n.N^*:{\cal T}(\tau_1) \rightarrow \dotsb \rightarrow {\cal T}(\tau_n) \rightarrow {\cal T}(\sigma) $$
The result follows by the definition of ${\cal T}$. The case where $x$ does not occur in $N$ is similar, with an extra application of the Weakening rule, thus we will omit it.
\item $M$ is of the form $M_1M_2$. In this case:
\begin{enumerate}
\item ${\cal E}_I(M_1:\tau_1 \cap \dotsb \cap \tau_n \rightarrow \sigma) \lhd (M_0,A_0)$
\item ${\cal E}_I(M_2:\tau_i) \lhd (N_i,A_i) (1 \leq i \leq n)$
\end{enumerate}
By the induction hypothesis 
\begin{enumerate}
\item ${\cal T}_e(A_0) \vdash_C M_0:{\cal T}(\tau_1) \rightarrow \dotsb \rightarrow {\cal T}(\tau_n) \rightarrow {\cal T}(\sigma)$
\item ${\cal T}_e(A_i) \vdash_C N_i:{\cal T}(\tau_i)$, for $(1 \leq i \leq n)$.
\end{enumerate}
Notice that the variables in $A_0, \ldots ,A_n$ are all distinct, because expansion contexts are generated with expansions of occurrences of free variables such that if a variable occurs in different expansions contexts it must occur with the same type, otherwise it was not possible to define the resulting expansion context $A_0 \uplus A_1 \uplus \dotsb \uplus A_n$. Thus the same happens for ${\cal T}_e(A_0), \ldots, {\cal T}_e(A_n)$. This guarantees that in ${\cal T}_e(A_0) \cup \dotsb \cup {\cal T}_e(A_n)$ all variables are distinct. Thus
$${\cal T}_e(A_0) \cup \dotsb \cup {\cal T}_e(A_n) \vdash_C M_0N_1 \ldots N_n:{\cal T}(\sigma)$$  
\end{enumerate}
\end{enumerate}
\end{proof}

\subsub*

\begin{proof} The proof will follow by structural induction on
$M$. Notice that, by Lemma \ref{times}, $x$ occurs free in $M$.
\begin{enumerate}
\item Base case: $M \equiv x$. In this case:
$${\cal E}_I(x:\sigma) \lhd (y,\{x:\{y:\sigma\}\})$$
thus
$${\cal E}_I(x[N/x]:\sigma) \lhd (y[N_1/y],A_1)$$
where
$${\cal E}_I(N:\sigma) \lhd (N_1,A_1)$$
\item Induction step: 
\begin{enumerate}
\item $M \equiv \lambda y.M_0$. Assume that $y$ occurs free in $M_0$. The other case is simpler. In this case:
$${\cal E}_I(\lambda y.M_0:\delta_1 \cap \dotsb \cap \delta_n
\rightarrow \sigma_1) \lhd (\lambda y_1 \ldots y_n.M_0^*, A_0 \uplus
\{x:\{x_1:\tau_1, \dots, x_k:\tau_k\}\})$$
by the definition of expansion we have
$${\cal E}_I(M_0:\sigma_1) \lhd (M_0^*, A_0 \uplus
\{x:\{x_1:\tau_1, \ldots, x_k:\tau_k\}\} \uplus \{y:\{y_1:\delta_1,
\ldots, y_n:\delta_n\}\}$$
and
$${\cal E}_I(N:\tau_i) \lhd (N_i,A_i), i \in \{1, \ldots, k\}$$
By the induction hypothesis it follows:
$${\cal E}_I(M_0[N/x]:\sigma_1) \lhd (M_0^*[N_1/x_1, \ldots,
  N_k/x_k],A_0 \uplus \dotsb \uplus A_k \uplus \{y:\{y_1:\delta_1,
\ldots, y_n:\delta_n\}\}$$
thus
$${\cal E}_I (\lambda y.M_0)[N/x]:\delta_1 \cap \dotsb \cap \delta_n
\rightarrow \sigma_1) \lhd (\lambda y_1 \ldots y_n.M_0^*[N_1/x_1, \ldots,
  N_k/x_k], A_0 \uplus \dotsb \uplus A_k)$$
\item $M \equiv M_1M_2$.
$${\cal E}_I(M_1M_2:\sigma) \lhd (P_0P_1 \ldots P_n,B_0 \uplus B_1
\uplus \dotsb \uplus B_n \uplus \{x:\{x_1:\tau_1, \ldots , x_k:\tau_k\}\})$$
and
$${\cal E}_I(N:\tau_i) \lhd (N_i,A_i), i \in \{1,\ldots,k\}$$
Let $X = \{x:\{x_1:\tau_1, \ldots, x_k:\tau_k\}\} = X_0 \uplus \dotsb \uplus
X_n$ where $X_i = \{x:\{x^i_1:\tau^i_1, \ldots,
x^i_{k_i}:\tau^i_{k_i}\}\}$ and $\{x^i_1,\ldots,x^i_{k_i}\}$ is the subset
of $\{x_1, \ldots ,x_k\}$ whose elements occur in $P_i$  for $i \in
\{0,\ldots,n\}$.
Now we have
$${\cal E}_I(M_1:\delta_1 \cap \dotsb \cap \delta_n \rightarrow \sigma) \lhd
(P_0,B_0 \uplus X_0)$$
and
$${\cal E}_I(M_2:\delta_i) \lhd (P_i, B_i \uplus X_i), i \in
\{1,\ldots,n\}$$
Let $T =\{N_1,\ldots,N_k\} = T_0 \cup \dotsb \cup T_n$ where $T_i = \{N^i_1,\ldots,N^i_{k_i}\}$ is the subset of $T$ whose elements occur
in $P_i[N_1/x_1,\ldots,N_k/x_k]$ and such that
${\cal E}_I(N:\tau^i_j) \lhd (N^i_j,A_j)$ for $i \in \{1,\ldots,n\}$ and
$j \in \{1, \ldots, k_i\}$.
By the induction hypothesis we have:
$${\cal E}_I(M_1[N/x]:\delta_1 \cap \ldots \cap \delta_n \rightarrow
\delta) \lhd (P_0[x^0_1/N^0_1,\ldots,x^0_{k_0}/N^0_{k_0}],B_0)$$
and
$${\cal E}_I(M_2[N/x]:\delta_i) \lhd (P_i[x^i_1/N^i_1,\ldots,x^i_{k_i}/N^i_{k_i}],B_i), i
\in \{1,\ldots,n\})$$
Thus
$${\cal E}_I((M_1M_2)[N/x]:\sigma) \lhd ((P_0P_1 \dotsb P_n)[x_1/N_1, \ldots,x_k/N_k],B_0 \uplus B_1 \uplus \ldots \uplus B_n)$$
\end{enumerate}
\end{enumerate}
\end{proof}

\sub*

\begin{proof} We will consider two cases:
\begin{enumerate} 
\item $x \in FV(M)$. By the definition of expansion:
$${\cal E}_I((\lambda x.M)N:\sigma) \lhd ((\lambda x_1 \ldots x_k.M_0)N_1
\ldots N_k,A_0 \uplus \dotsb \uplus A_k)$$
where
$${\cal E}_I((\lambda x.M):\tau_1 \cap \dotsb \cap \tau_k \rightarrow
\sigma) \lhd (\lambda x_1 \ldots x_k.M_0,A_0)$$ 
and
$${\cal E}_I(N:\tau_i) \lhd (N_i,A_i)$$
Then we have
$${\cal E}_I(M:\sigma) \lhd (M_0,A_0 \uplus \{x:\{x_1:\tau_1, \ldots,
x_k:\tau_k\}\})$$
By Lemma \ref{sub_sub} we have:
$${\cal E}_I(M[N/x]:\sigma) \lhd (M_0[N_1/x_1, \ldots, N_k/x_k],A_0
\uplus \dotsb \uplus A_k)$$
\item $x \not \in FV(M)$. In
this case
$${\cal E}_I((\lambda x.M)N:\sigma) \lhd ((\lambda y.M_0)N_0,A_0 \uplus
A_1)$$
where
$${\cal E}_I(\lambda x.M:\tau \rightarrow \sigma) \lhd (\lambda
y.M_0,A_0)$$
and
$${\cal E}_I(N:\tau) \lhd (N_0,A_1)$$
Thus
$${\cal E}_I(M:\sigma) \lhd (M_0,A_0)$$
and
$${\cal E}_I(M[N/x]:\sigma) = {\cal E}(M:\sigma) \lhd (M_0,A_0)$$
Note that $A_0 \sqsubseteq A_0 \uplus A_1.$ 
\end{enumerate}
\end{proof}

To show that expansion is preserved by weak head reduction we need the
concept of {\em context} as a term containing one hole $[\  ]$.
\begin{definition}
{\em Contexts} $C[\ ]$ are described by:
\begin{enumerate}
\item $[\ ]$ is a context;
\item If $C[\ ]$ is a context and $M$ a $\lambda$-term, then $C[\ ]M$,
  $MC[\ ]$ and $\lambda x.C[\ ]$ are contexts.
\end{enumerate}
\end{definition}
If $M$ is a $\lambda$-term and $C[\ ]$ a context then $C[M]$ is the
result of replacing the hole in $C[\ ]$ with $M$. Note that this
operation is different from that of substitution because no renaming
of bound variables is allowed.

\tfive*

\begin{proof} We will consider two cases:
\begin{enumerate} 
\item $x \in FV(M)$. By the definition of expansion:
$${\cal E}_I((\lambda x.M)N:\sigma) \lhd ((\lambda x_1 \ldots x_k.M_0)N_1
\ldots N_k,A_0 \uplus \dotsb \uplus A_k)$$
where
$${\cal E}_I((\lambda x.M):\tau_1 \cap \dotsb \cap \tau_k \rightarrow
\sigma) \lhd (\lambda x_1 \ldots x_k.M_0,A_0)$$ 
and
$${\cal E}_I(N:\tau_i) \lhd (N_i,A_i)$$
Then we have
$${\cal E}_I(M:\sigma) \lhd (M_0,A_0 \uplus \{x:\{x_1:\tau_1, \ldots,
x_k:\tau_k\}\})$$
By Lemma \ref{sub_sub} we have:
$${\cal E}_I(M[N/x]:\sigma) \lhd (M_0[N_1/x_1, \ldots, N_k/x_k],A_0
\uplus \dotsb \uplus A_k)$$
\item $x \not \in FV(M)$. In
this case:
$${\cal E}_I((\lambda x.M)N:\sigma) \lhd ((\lambda y.M_0)N_0,A_0 \uplus
A_1)$$
where
$${\cal E}_I(\lambda x.M:\tau \rightarrow \sigma) \lhd (\lambda
y.M_0,A_0)$$
and
$${\cal E}_I(N:\tau) \lhd (N_0,A_1)$$
Thus
$${\cal E}_I(M:\sigma) \lhd (M_0,A_0)$$
and
$${\cal E}_I(M[N/x]:\sigma) = {\cal E}(M:\sigma) \lhd (M_0,A_0)$$
Note that $A_0 \sqsubseteq A_0 \uplus A_1.$ 
\end{enumerate}
\end{proof}

\red*

\begin{proof} We will use structural induction on the context $C$
such that $M_1
 \overset{R}{\underset{\beta}{\rightarrow}} M_2$ and $M_1 \equiv
 C[R]$.
\begin{enumerate}
\item Base case: $M_1$ is the $\beta$-redex $R$. The proof follows
  from Lemma \ref{subI}.
\item Induction step: 
\begin{enumerate}
\item $M_1 = \lambda x.C[R]$. In this case:
$$ {\cal E}_I(\lambda x.C[R]:\delta_1 \cap \dotsb \cap \delta_k
\rightarrow \delta) \lhd (\lambda x_1 \ldots x_k.M^*,A)$$
$x$ occurs free in $\lambda x.C[R]$, thus
$${\cal E}_I(C[R]:\delta) \lhd (M^*,A \uplus
\{x:\{x_1:\delta_1,\ldots,x_k:\delta_k\}\})$$
By the induction hypothesis there is a term $N_2$ such that
$$C[R] \underset{\beta}{\rightarrow} P$$
and
$${\cal E}_I(P:\delta) \lhd (N_2,A \uplus \{x:\{x_1:\delta_1, \ldots,
x_k:\delta_k\}\}$$
and
$$M^* \underset{\beta}{\twoheadrightarrow} N_2 $$
Thus
$$\lambda x.C[R] \underset{\beta}{\rightarrow} \lambda x.P$$
$${\cal E}_I(\lambda x.P: \delta_1 \cap \dotsb \cap \delta_k \rightarrow
\delta) \lhd (\lambda x_1 \ldots x_k.N_2,A)$$
and
$$\lambda x_1 \ldots x_k.M^* \underset{\beta}{\twoheadrightarrow}
\lambda x_1 \ldots x_k.N_2$$
\item $M_1 \equiv C[R]W$. In this case
$${\cal E}_I(C[R]W:\sigma) \lhd (P_0P_1 \dotsb P_k,A_0 \uplus A_1 \uplus
\ldots \uplus A_k)$$
Thus
$${\cal E}_I(C[R]:\tau_1 \cap \dotsb \cap \tau_k \rightarrow \sigma)
\lhd (P_0,A_0)$$
and
$${\cal E}_I(W:\tau_i) \lhd (P_i,A_i)$$
By the induction hypothesis there is a term $P_0^*$ such that
$$C[R] \underset{\beta}{\rightarrow} P$$
$${\cal E}(P:\tau_1 \cap \dotsb \cap \tau_k \rightarrow \sigma) \lhd
(P_0^*,A_0)$$ 
and
$$P_0 \underset{\beta}{\twoheadrightarrow} P_0^*$$
Thus
$${\cal E}_I(PW:\sigma) \lhd (P_0^* P_1 \ldots P_k,A_0 \uplus \dotsb \uplus
A_k)$$ and
$$P_0P_1\ldots P_k \underset{\beta}{\twoheadrightarrow} P_0^*P_1\ldots
P_k$$
\item Suppose that $M_1C[R] \underset{\beta}{\rightarrow}
  M_1N_2$. Thus
$${\cal E}_I(M_1C[R]) \lhd (P_0P_1 \ldots P_k,A_0 \uplus \ldots \uplus
A_k)$$
Thus
$${\cal E}_I(M_1:\tau_1 \cap \ldots \cap \tau_k \rightarrow \sigma) \lhd
(P_0,A_0)$$
and
$${\cal E}_I(C[R]:\tau_i) \lhd (P_i,A_i), i \in \{1,\ldots,k\}$$
By the induction hypothesis for $i \in \{1,\ldots,k\}$ there are terms
$P_i^*$ such that
$${\cal E}_I(N_2:\tau_i) \lhd (P_i^*,A_i)$$
and
$$P_i \underset{\beta}{\twoheadrightarrow} P_i^*$$
Thus
$${\cal E}_I(M_1N_2) \lhd (P_0P_1^* \ldots P_k^*,A_0 \uplus \dotsb \uplus
A_k)$$
and
$$P_0P_1 \ldots P_k \underset{\beta}{\twoheadrightarrow} P_0P_1^*
\ldots P_k^* $$
\end{enumerate}
\end{enumerate}
\end{proof}

\orderexp*

\begin{proof}By induction on $M$. 
\begin{itemize}
    \item $\expo(x : \sigma) \goto (y^{\at{\sigma}}, \octx{x: \octx{y: \at{\sig}}})$ and  $\cetct{\octx{x: \octx{y: \at{\sig}}}} = y: \at{\sig}$ and trivially we have $y: \at{\sig} \vdash_o y: \at{\sig}$.
    \item $\expo(\l{x}.M : \sig_1 \cap \dotsb \cap \sig_n \llto \sig) \goto  (\l{x_1 \dots x_n}.M_0^{\at{\sig_1} \llto_{r} \dotsb \llto_{r} \at{\sig_n} \llto_r \at{\sigma}}, A)$, follows from $\expo(M : \sig) \goto (M_0^{\at{\sigma}},\ A+\octx{x : \octx{x_1:\at{\sig_1},\dots,x_n:\at{\sig_n}}})$. By the \ih $$\cetct{A},x_1:\at{\sig_1},\dots,x_n:\at{\sig_n}\vdash_o M_0: \at{\sigma}$$
    Applying $\introarrowr$ $n$ times, one gets:
    $$\cetct{A}\vdash_o \l{x_1 \dots x_n}.M_0: \at{\sig_1} \llto_{r} \dotsb \llto_{r} \at{\sig_n} \llto_r \at{\sigma}$$
    \item $\expo(\l{x}.M : \sig_1 \cap \dotsb \cap \sig_n \llto \sig) \goto  (\l{x_1 \dots x_n}.M_0^{\at{\sig_1} \llto_{l} \dotsb \llto_{l} \at{\sig_n} \llto_l \at{\sigma}}, A)$, follows from $\expo(M : \sig) \goto (M_0^{\at{\sigma}},\octx{x : \octx{x_n:\at{\sig_n},\dots,x_1:\at{\sig_1}}}+A)$. By the \ih $$x_n:\at{\sig_n},\dots,x_1:\at{\sig_1},\cetct{A}\vdash_o M_0: \at{\sigma}$$
    Applying $\introarrowl$ $n$ times, one gets:
    $$\cetct{A}\vdash_o \l{x_1 \dots x_n}.M_0: \at{\sig_1} \llto_{l} \dotsb \llto_{l} \at{\sig_n} \llto_l \at{\sigma}$$
    \item $\expo(M N : \sig)  \goto ((M_0 N_1 \dots N_m)^{\at{\sig}},  A_0+ A_1+ \dotsb+ A_m)$, follows from $\expo(M : \sig_1 \cap \dots \cap \sig_m \llto \sig) \goto (M_0^{\at{\sig_1} \llto_r \dotsb \llto_r \at{\sig_m} \llto_r \at{\sig}}, A_0)$ and $\left( \expo(N : \sig_i) \goto (N_i^{\at{\sig_i}}, A_i) \right)_{\iom}$, for some $m > 0$ and $\sig_1,\dots,\sig_m$. By the \ih 
    $$\cetct{A_0} \vdash_o M_0:\at{\sig_1} \llto_r \dotsb \llto_r \at{\sig_m} \llto_r \at{\sig}$$
    and $$\left(\cetct{A_i} \vdash_o N_i: \at{\sig_i}\right)_{\iom}$$
     Applying $\elimarrowr$ $m$ times, one gets:
     $$\cetct{A_0},\cetct{A_1},\dots,\cetct{A_m} \vdash_o M_0N_1\dots N_m: \at{\sig}$$
     \item $\expo(M N : \sig)  \goto ((M_0 N_1 \dots N_m)^{\at{\sig}}, A_m+ \dotsb+ A_1+A_0)$, follows from $\expo(M : \sig_1 \cap \dots \cap \sig_m \llto \sig) \goto (M_0^{\at{\sig_1} \llto_l \dotsb \llto_l \at{\sig_m} \llto_l \at{\sig}}, A_0)$ and $\left( \expo(N : \sig_i) \goto (N_i^{\at{\sig_i}}, A_i) \right)_{\iom}$, for some $m > 0$ and $\sig_1,\dots,\sig_m$. By the \ih 
    $$\cetct{A_0} \vdash_o M_0:\at{\sig_1} \llto_l \dotsb \llto_l \at{\sig_m} \llto_l \at{\sig}$$
    and $$\left(\cetct{A_i} \vdash_o N_i: \at{\sig_i}\right)_{\iom}$$
     Applying $\elimarrowl$ $m$ times, one gets:
     $$\cetct{A_m},\dots,\cetct{A_1},\cetct{A_0} \vdash_o M_0N_1\dots N_m: \at{\sig}$$
\end{itemize}
\end{proof}

\occur*
\begin{proof}
By induction on $M$.
\end{proof}
\subst*

\begin{proof}
By induction on $M$.
\begin{itemize}
    \item $M=x$: $\expo(M,\sig) \goto (y^{\at{\sig}},\octx{x:\octx{y:\at{\sig}}})$ with $y$ fresh and $\expo(N,\sig)\goto(N_1^{\at{\sig}},A_1)$. Then $\expo(x[N/x],\sig) = \expo(N,\sig) = (N_1^{\at{\sig}},A_1) = (y[N_1/y,A_1)$. Note that in this case $A_0=A_{n+1}=\varnothing$.
    \item $M=\lambda y.M'$, we have two cases:
    \begin{enumerate}
        \item $\expo(\l{y}.M' : \sig_1 \cap \dotsb \cap \sig_k \llto \sig)  \goto (\l{y_1 \dots y_k}.M_0^{\at{\sig_1} \llto_{r} \dotsb \llto_{r} \at{\sig_k} \llto_r \at{\sigma}}, A_0+\octx{x:\octx{x_1:\at{\tau_1},\dots,x_n:\at{\tau_n}}}+A_{n+1})$ if $y \in \fv{M'}$  and $$\expo(M' : \sig) \goto (M_0^{\at{\sig}}, A_0+\octx{x:\octx{x_1:\at{\tau_1},\dots,x_n:\at{\tau_n}}}+A_{n+1} + \octx{y:\octx{y_1:\at{\sig_1},\dots,y_k:\at{\sig_k}}}).$$ 
        By the \ih, $$\expo(M'[N/x],\sig) \goto (M_0[N_1/x_1,\dots,N_n/x_n]^{\at{\sig}},A_0+_{\ion} A_i + A_{n+1}+\octx{x:\octx{y_1:\at{\sig_1},\dots,y_k:\at{\sig_k}}}),$$ from which follows $$\expo((\ y.M')[N/x], \sig_1 \cap \dotsb \cap \sig_k \llto \sig) \goto ((\l y_1 \dots y_k. M_0[N_1/x_1,\dots,N_n/x_n])^{\at{\sig_1} \llto_{r} \dotsb \llto_{r} \at{\sig_k} \llto_r \at{\sigma}},A)$$
        where $A=A_0+_{\ion} A_i + A_{n+1}$.
        \item $\expo(\l{y}.M' : \sig_1 \cap \dotsb \cap \sig_k \llto \sig)  \goto (\l{y_1 \dots y_k}.M_0^{\at{\sig_1} \llto_{l} \dotsb \llto_{l} \at{\sig_k} \llto_l \at{\sigma}}, A_0+\octx{x:\octx{x_1:\at{\tau_1},\dots,x_n:\at{\tau_n}}}+A_{n+1})$ if $y \in \fv{M'}$ and $$\expo(M' : \sig) \goto (M_0^{\at{\sig}}, \octx{y:\octx{y_k:\at{\sig_k},\dots,y_1:\at{\sig_1}}}+ A_0+\octx{x:\octx{x_1:\at{\tau_1},\dots,x_n:\at{\tau_n}}}+A_{n+1}  ).$$
        By the \ih, $$\expo(M'[N/x],\sig) \goto (M_0[N_1/x_1,\dots,N_n/x_n]^{\at{\sig}},\octx{y:\octx{y_k:\at{\sig_k},\dots,y_1:\at{\sig_1}}} + A_0+_{\ion} A_i + A_{n+1}),$$ from which follows $$\expo((\l y.M')[N/x], \sig_1 \cap \dotsb \cap \sig_k \llto \sig) \goto ((\l y_1 \dots y_k. M_0[N_1/x_1,\dots,N_n/x_n])^{\at{\sig_1} \llto_{l} \dotsb \llto_{l} \at{\sig_k} \llto_l \at{\sigma}},A),$$
        where $A=A_0+_{\ion} A_i + A_{n+1}$.
    \end{enumerate}
    \item $M=PQ$, we have two cases:
        \begin{enumerate}
            \item $\expo(PQ : \sig) \goto ((P_0 Q_1 \dots Q_m)^{\at{\sig}},  A_0+\octx{x:\octx{x_1:\at{\tau_1},\dots,x_n:\at{\tau_n}}}+A_{n+1})$,
            if for some $m > 0$ and $\sig_1,\dots,\sig_m$, $$\expo(P : \sig_1 \cap \dots \cap \sig_m \llto \sig) \goto (P_0^{\at{\sig_1} \llto_r \dotsb \llto_r \at{\sig_m} \llto_r \at{\sig}}, B_0)$$  and  $$\left( \expo(N : \sig_j) \goto (N_j^{\at{\sig_j}}, B_j) \right)_{\jom},$$ where $A_0+\octx{x:\octx{x_1:\at{\tau_1},\dots,x_n:\at{\tau_n}}}+A_{n+1} = B_0+B_1+\dotsb+ B_m$. 
            Let \begin{itemize}
                \item $\{x_1,\dots,x_n\} = \{x_{01},\dots, x_{0i_0}, x_{11},\dots, x_{1i_1},\dots,x_{m1},\dots,x_{mi_m}\}$,
                \item $B_0 = A_0' + \underbrace{\octx{x:\octx{x_{01}:\at{\tau_{01}},\dots, x_{0i_0}:\at{\tau_{0i_0}}}}}_{A_0''} + A_0'''$,  $B_j = B_j' + \underbrace{\octx{x:\octx{x_{j1}:\at{\tau_{j1}},\dots, x_{ji_j}:\at{\tau_{ji_j}}}}}_{B_j''} + B_j'''$, for $\jom$, 
                \item $A_0 = A_0'$, $A_{n+1} = A_0''' +_{\jom} (B_j'+B_j''')$ 
                \item and $A_{1\dots n} = A_0'' +_{k=1\dots n} B_k''$. 
            \end{itemize}
            By the \ih, 
            $$\expo(P[N/x]: \sig_1 \cap \dots \cap \sig_n \llto \sig) \goto ((P_0[N_{01}/x_{01},\dots,N_{0i_0}/x_{0i_0}])^{\at{\sig_1} \llto_r \dotsb \llto_r \at{\sig_n} \llto_r \at{\sig}},A_0'+A_{01}+\dotsb+A_{0i_0} + A_0''')$$
            and, for $\jom$,
             $$\expo(Q[N/x]: \sig_j) \goto ((Q_j[N_{j1}/x_{j1},\dots,N_{ji_j}/x_{ji_j}])^{\at{\sig_j}},B_j'+A_{j1}+\dotsb+A_{ji_j} + B_j''')$$
             from which follows:
             $$\expo(PQ[N/x]: \sig) \goto ((P_0Q_1\dots Q_m[N_{01}/x_{01},\dots N_{j1}/x_{ji_j}])^{\at{\sig}},A_0 + A_1+\dotsb+A_n +A_{n+1})$$
             where $A_0 + A_1+\dotsb+A_n +A_{n+1}= A_0'+A_{01}+\dotsb+A_{0i_0} + A_0'''+_{\jom} (B_j'+A_{j1}+\dotsb+A_{ji_j} + B_j''') $  given that  $A_0' = A_0$, $A_{01}+\dotsb+A_{0i_0}+\dotsb +A_{m1}+\dotsb+A_{mi_m} = A_1+\dotsb+A_n$ and $A_{n+1} = A_0''' +_{\jom} (Bj'+Bj''')$. 
             \item $\expo(PQ : \sig) \goto ((P_0 Q_1 \dots Q_m)^{\at{\sig}},  A_0+\octx{x:\octx{x_1:\at{\tau_1},\dots,x_n:\at{\tau_n}}}+A_{n+1})$,
            if for some $m > 0$ and $\sig_1,\dots,\sig_m$  $$\expo(P : \sig_1 \cap \dots \cap \sig_m \llto \sig) \goto (P_0^{\at{\sig_1} \llto_l \dotsb \llto_l \at{\sig_m} \llto_l \at{\sig}}, B_0)$$  and  $$\left( \expo(N : \sig_j) \goto (N_j^{\at{\sig_j}}, B_j) \right)_{\jom},$$ where $A_0+\octx{x:\octx{x_1:\at{\tau_1},\dots,x_n:\at{\tau_n}}}+A_{n+1} = B_m+\dotsb+ B_1+B_0$. Let 
            \begin{itemize}
                \item $\{x_1,\dots,x_n\} = \{x_{m1},\dots x_{mi_m}, \dots,x_{01},\dots x_{0i_0}\}$,
                \item $B_0 = A''' + \underbrace{\octx{x:\octx{x_{01}:\at{\tau_{01}},\dots x_{0i_0}:\at{\tau_{0i_0}}}}}_{A_0''} + A_0'$,  $B_j = B_j''' + \underbrace{\octx{x:\octx{x_{j1}:\at{\tau_{j1}},\dots x_{ji_j}:\at{\tau_{ji_j}}}}}_{B_j''} + B_j'$ for $\jom$,
                \item  $A_{n+1} = A_0'$, $A_0 =+_{j=m\dots 1} (B_j'+B_j''')+ A_0'''$,
                \item and $A_{1\dots n} =  +_{k=n\dots 1} B_k''+ A_0''$.
            \end{itemize}
            
             By the \ih, 
            $$\expo(P[N/x]: \sig_1 \cap \dots \cap \sig_n \llto \sig) \goto ((P_0[N_{01}/x_{01},\dots,N_{0i_0}/x_{0i_0}])^{\at{\sig_1} \llto_l \dotsb \llto_l \at{\sig_n} \llto_l \at{\sig}},A_0'''+A_{01}+\dotsb+A_{0i_0} + A_0')$$
            and, for $\jom$,
             $$\expo(Q[N/x]: \sig_j) \goto ((Q_j[N_{j1}/x_{j1},\dots,N_{ji_j}/x_{ji_j}])^{\at{\sig_j}},B_j'''+A_{j1}+\dotsb+A_{ji_j} + B_j')$$
             from which follows:
             $$\expo(PQ[N/x]: \sig) \goto ((P_0Q_1\dots Q_m[N_{01}/x_{01},\dots N_{j1}/x_{ji_j}])^{\at{\sig}},A_0 + A_1+\dotsb+A_n +A_{n+1}$$
             
             where $A_0 + A_1+\dotsb+A_n +A_{n+1} = +_{j=m\dots 1} (B_j'''+A_{j1}+\dotsb+A_{ji_j} + B_j') A_0'''+A_{01}+\dotsb+A_{0i_0} + A_0')$.
        \end{enumerate}
    
\end{itemize}
\end{proof}

\reduction*
\begin{proof}
We consider two cases:
\begin{enumerate}
    \item $\expo((\lambda x.M)N:\sig) \goto ((\lambda x_1\dots x_n.M^*)N_1\dots N_n, +_{\izn} A_i)$ follows from $\expo(M,\sig) \goto (M^*, A_0 + \octx{x:\octx{x_1:\at{\tau_1},\dots,x_n:\at{\tau_n}}})$ and $\expo(N,\tau_i)\goto(N_i,A_i)$ for some $\tau_1,\dots,\tau_n$, for $\izn$. Then, from Lemma~\ref{lem:subst}, it follows $\expo(M[N/x]:\sig) \goto (M^*[N_1/x_1,\dots,N_n/x_n], +_{\izn} A_i)$ and $$(\lambda x_1\dots x_n.M^*)N_1\dots N_n \underset{\beta}{\twoheadrightarrow} M^*[N_1/x_1,\dots,N_n/x_n]. $$
    \item $\expo((\lambda x.M)N:\sig) \goto ((\lambda x_1\dots x_n.M^*)N_1\dots N_n, A_n+\dots+A_1+A_0)$ follows from $\expo(M,\sig) \goto (M^*, \octx{x:\octx{x_n:\at{\tau_n},\dots,x_1:\at{\tau_1}}}+ A_0)$ and $\expo(N,\tau_i)\goto(N_i,A_i)$ for some $\tau_1,\dots,\tau_n$, for $\izn$. Then, from Lemma~\ref{lem:subst}, it follows $\expo(M[N/x]:\sig) \goto (M^*[N_n/x_n,\dots,N_1/x_1], A_n+\dots+A_1+A_0)$ and $$(\lambda x_1\dots x_n.M^*)N_1\dots N_n \underset{\beta}{\twoheadrightarrow} M^*[N_1/x_1,\dots,N_n/x_n] \equiv M^*[N_n/x_n,\dots,N_1/x_1]$$
\end{enumerate}
\end{proof}
\icalculus*
\begin{proof}
 We will use structural induction on the context $C[R]$ such that $M_1 \overset{R}{\underset{\beta}{\rightarrow}} M_2$ and $M_1 \equiv C[R]$.
\begin{enumerate}
\item Base case: $M_1$ is the $\beta$-redex $R$. The proof follows
  from Lemma~\ref{lem:reduction}.
\item Induction step: 
\begin{itemize}
\item $M_1 \equiv \lambda x.C[R] \overset{R}{\underset{\beta}{\rightarrow}} \lambda x.P \equiv M_2$ with $C[R] \overset{R}{\underset{\beta}{\rightarrow}} P$. We have two cases:
\begin{enumerate}
    \item $$\expo(\lambda x.C[R]:\sig_1 \cap \dotsb \cap \sig_n
\rightarrow\sig) \goto ((\lambda x_1 \dotsb x_n.M^*)^{\at{\sig_1} \llto_r \dotsb \llto_r \at{\sig_n} \llto_r \at{\sig}},A),$$ which follows from $$\expo(C[R]:\sig) \goto ({M^*}^{\at{\sig}},A+\octx{x:\octx{x_1:\at{\sig_1},\dots, x_n: \at{\sig_n}}}).$$  By the \ih, there is a term $N$ such that $$\expo(P:\sig) \goto({N}^{\at{\sig}},A+\octx{x:\octx{x_1:\at{\sig_1},\dots, x_n: \at{\sig_n}}})$$ with $M^* \underset{\beta}{\twoheadrightarrow} N$. From which follows $$\expo(\lambda x. P:\sig_1 \cap \dotsb \cap \sig_n
\rightarrow \sig) \goto((\lambda x_1 \ldots x_n.N)^{\at{\sig_1} \llto_r \dotsb \llto_r \at{\sig_n} \llto_r \at{\sig}},A)$$ and $\lambda x_1 \ldots x_n.M^* \underset{\beta}{\twoheadrightarrow}
\lambda x_1 \ldots x_n.N$, as expected.
  \item $$\expo(\lambda x.C[R]:\sig_1 \cap \dotsb \cap \sig_n
\rightarrow\sig) \goto ((\lambda x_1 \ldots x_n.M^*)^{\at{\sig_1} \llto_l \dotsb \llto_l \at{\sig_n} \llto_l \at{\sig}},A),$$ which follows from $$\expo(C[R]:\sig) \goto ({M^*}^{\at{\sig}},\octx{x:\octx{x_n:\at{\sig_n},\dots, x_1: \at{\sig_1}}}+A).$$  By the \ih, there is a term $N$ such that $$\expo(P:\sig) \goto({N}^{\at{\sig}},\octx{x:\octx{x_n:\at{\sig_n},\dots, x_1: \at{\sig_1}}}+A)$$ with $M^* \underset{\beta}{\twoheadrightarrow} N $. From which follows $$\expo(\lambda x. P:\sig_1 \cap \dotsb \cap \sig_n
\rightarrow \sig) \goto((\lambda x_1 \ldots x_n.N)^{\at{\sig_1} \llto_l \dotsb \llto_l \at{\sig_n} \llto_l \at{\sig}},A)$$ and $\lambda x_1 \ldots x_n.M^* \underset{\beta}{\twoheadrightarrow} \lambda x_1 \ldots x_n.N$, as expected.
\end{enumerate}
\item $M_1 \equiv C[R]W
\overset{R}{\underset{\beta}{\rightarrow}} PW \equiv M_2$ with $C[R] \overset{R}{\underset{\beta}{\rightarrow}} P$. We have two cases:
\begin{enumerate}
    \item $$\expo(C[R]W:\sig) \goto (P_0P_1\dots P_n, +_{\izn} A_i),$$ which follows from $$\expo(C[R]:\sig_1 \cap \dotsb \cap \sig_n \rightarrow \sig) \goto(P_0^{\at{\sig_1} \llto_r \dotsb \llto_r \at{\sig_n} \llto_r \at{\sig}}, A_0)$$ and  $$\expo(W:\sig_i) \goto (P_i, A_i)$$ for some $\sig_1,\dots,\sig_n$ for $\ion$. By the \ih there is a term $N$ such that $$\expo(P:\sig_1 \cap \dotsb \cap \sig_n \rightarrow \sig) \goto(N^{\at{\sig_1} \llto_r \dotsb \llto_r \at{\sig_n} \llto_r \at{\sig}}, A_0)$$ and $P_0 \underset{\beta}{\twoheadrightarrow} N$, from which follows $$\expo(PW:\sig) \goto (NP_1\dots P_n, +_{\izn} A_i)$$ and $P_0P_1\dots P_n \underset{\beta}{\twoheadrightarrow} NP_1\dots P_n$, as expected.
     \item $$\expo(C[R]W:\sig) \goto (P_0P_1\dots P_n, A_n+\dotsb+A_1+A_0),$$ which follows from $$\expo(C[R]:\sig_1 \cap \dotsb \cap \sig_n \rightarrow \sig) \goto(P_0^{\at{\sig_1} \llto_l \dotsb \llto_l \at{\sig_n} \llto_l \at{\sig}}, A_0)$$ and  $$\expo(W:\sig_i) \goto (P_i, A_i)$$ for some $\sig_1,\dots,\sig_n$ for $\ion$. By the \ih there is a term $N$ such that $$\expo(P:\sig_1 \cap \dotsb \cap \sig_n \rightarrow \sig) \goto(N^{\at{\sig_1} \llto_l \dotsb \llto_l \at{\sig_n} \llto_l \at{\sig}}, A_0)$$ and $P_0 \underset{\beta}{\twoheadrightarrow} N$, from which follows $$\expo(PW:\sig) \goto (NP_1\dots P_n, A_n+\dots+A_1+A_0)$$ and $P_0P_1\dots P_n \underset{\beta}{\twoheadrightarrow} NP_1\dots P_n$, as expected.
\end{enumerate}
\item $M_1 \equiv WC[R]
\overset{R}{\underset{\beta}{\rightarrow}} WP \equiv M_2$ with $C[R] \overset{R}{\underset{\beta}{\rightarrow}} P$. We have two cases:
\begin{enumerate}
    \item $$\expo(WC[R]:\sig) \goto (P_0P_1\dots P_n, +_{\izn} A_i),$$ which follows from $$\expo(P:\sig_1 \cap \dotsb \cap \sig_n \rightarrow \sig) \goto(P_0^{\at{\sig_1} \llto_r \dotsb \llto_r \at{\sig_n} \llto_r \at{\sig}}, A_0)$$ and  $$\expo(C[R]:\sig_i) \goto (P_i, A_i)$$ for some $\sig_1,\dots,\sig_n$ for $\ion$. By the \ih there are terms $N_1 \dots,N_n$ such that $$\expo(P:\sig_i) \goto(N_i^{\at{\sig_i}} , A_i)$$ and $P_i \underset{\beta}{\twoheadrightarrow} N_i$, for $\ion$. Therefore $$\expo(WP:\sig) \goto (P_0N_1\dots N_n, +_{\izn} A_i)$$ and $P_0P_1\dots P_n \underset{\beta}{\twoheadrightarrow} P_0N_1\dots N_n$, as expected.
     \item $$\expo(WC[R]:\sig) \goto (P_0P_1\dots P_n,A_n+\dotsb+A_1+A_0),$$ which follows from $$\expo(P:\sig_1 \cap \dotsb \sig_n \rightarrow \sig) \goto(P_0^{\at{\sig_1} \llto_l \dotsb \llto_l \at{\sig_n} \llto_l \at{\sig}}, A_0)$$ and  $$\expo(C[R]:\sig_i) \goto (P_i, A_i)$$ for some $\sig_1,\dots,\sig_n$ for $\ion$. By the \ih there are terms $N_1 \dots,N_n$ such that $$\expo(P:\sig_i) \goto(N_i^{\at{\sig_i}} , A_i)$$ and $P_i \underset{\beta}{\twoheadrightarrow} N_i$, for $\ion$. Therefore $$\expo(WP:\sig) \goto (P_0N_1\dots N_n, A_n+\dotsb+ A_1+A_0)$$ and $P_0P_1\dots P_n \underset{\beta}{\twoheadrightarrow} P_0N_1\dots N_n$, as expected.
\end{enumerate}    
\end{itemize}
\end{enumerate}
\end{proof}